\documentclass[11pt]{article}
\pdfoutput=1
\usepackage{latexsym}
\usepackage{amsmath}
\usepackage{amssymb}
\usepackage{amsfonts,amsthm}
\usepackage{multicol}
\usepackage{tikz}
\usepackage{url}
\usepackage[top=1in, bottom=1in, left=1in, right=1in]{geometry}

\usepackage{thmtools, thm-restate}

\usepackage{tikz}
\usetikzlibrary{arrows,decorations.pathmorphing,decorations.shapes,backgrounds,fit}
\pgfkeys{tikz/cc1/.style={dotted} }
\pgfkeys{tikz/cc4/.style={dashed} }
\pgfkeys{tikz/cc3/.style={decorate,decoration=bumps} }
\pgfkeys{tikz/cc2/.style={decorate,decoration=snake} }
\pgfkeys{tikz/solid/.style={line width=2pt} }
\pgfkeys{tikz/solid2/.style={line width=1pt} }
\pgfkeys{tikz/nodostile1/.style={draw,line width=1pt}}
\pgfkeys{tikz/nodostile2/.style={draw,line width=2pt}}
\pgfkeys{tikz/arcostile1/.style={very thick} }
\pgfkeys{tikz/arcostile2/.style={red,very thick} }
\pgfkeys{tikz/arcostile3/.style={thick} }
\pgfkeys{tikz/arcostile4/.style={line width=2pt} }

\newcommand{\ignore}[1]{}

\newtheorem{theorem}{Theorem}[section]
\newtheorem{lemma}[theorem]{Lemma}

\usepackage{array}
\newcolumntype{C}[1]{>{\centering\let\newline\\\arraybackslash\hspace{0pt}}m{#1}}

\usepackage[lined, algoruled, linesnumbered]{algorithm2e}

\usepackage{hyperref,xcolor}
\definecolor{winered}{rgb}{0.5,0,0}

\hypersetup
{
    pdfborder={0 0 0},
    colorlinks=true,
    linkcolor={winered},
    urlcolor={winered},
    filecolor={winered},
    citecolor={winered},
    linktoc=all,
}

\begin{document}

\title{On Low-High Orders of Directed Graphs: An Incremental Algorithm and Applications}
\author{Loukas Georgiadis\footnotemark[1] \and Aikaterini Karanasiou\footnotemark[1] \and Giannis Konstantinos\footnotemark[1] \and Luigi Laura\footnotemark[2]}
\maketitle

\begin{abstract}
A flow graph $G=(V,E,s)$ is a directed graph with a distinguished start vertex $s$. The dominator tree $D$ of $G$ is a tree rooted at $s$, such that a vertex $v$ is an ancestor of a vertex $w$ if and only if all paths from $s$ to $w$ include $v$. The dominator tree is a central tool in program optimization and code generation, and has many applications in other diverse areas including constraint programming, circuit testing, biology, and in algorithms for graph connectivity problems. A low-high order of $G$ is a preorder $\delta$ of $D$ that certifies the correctness of $D$,
and has further applications in connectivity and path-determination problems. In this paper we first consider how to maintain efficiently a low-high order of a flow graph incrementally under edge insertions. We present algorithms that run in $O(mn)$ total time for a sequence of $m$ edge insertions in an initially empty flow graph with $n$ vertices.
These immediately provide the first incremental \emph{certifying algorithms} for maintaining the dominator tree in $O(mn)$ total time, and also imply incremental algorithms for other problems. Hence, we provide a substantial improvement over the $O(m^2)$ simple-minded algorithms, which recompute the solution from scratch after each edge insertion.
We also show how to apply low-high orders to obtain a linear-time $2$-approximation algorithm for the smallest $2$-vertex-connected spanning subgraph problem (2VCSS).
Finally, we present efficient implementations of our new algorithms for the incremental low-high and 2VCSS problems, and conduct an extensive experimental study on real-world graphs taken from a variety of application areas. The experimental results show that our algorithms perform very well in practice.
\end{abstract}

\footnotetext[1]{Department of Computer Science \& Engineering, University of Ioannina, Greece. E-mail: \texttt{loukas@cs.uoi.gr, akaranas@cs.uoi.gr, giannis\_konstantinos@outlook.com}.}
\footnotetext[2]{Dipartimento di Ingegneria Informatica, Automatica e Gestionale, ``Sapienza'' Universit\`a di Roma, Italy. E-mail: \texttt{laura@dis.uniroma1.it}.}
\thispagestyle{empty}

\section{Introduction}

A \emph{flow graph} $G=(V,E,s)$ is a directed graph (digraph) with a distinguished start vertex $s \in V$.
A vertex $v$ is \emph{reachable} in $G$ if there is a path from $s$ to $v$; $v$ is \emph{unreachable} if no such path exists. The \emph{dominator relation} in $G$ is defined for the set of reachable vertices as follows.
A vertex $v$ is a \emph{dominator} of a vertex $w$ ($v$ \emph{dominates} $w$) if every path from $s$ to $w$ contains $v$; $v$ is a \emph{proper dominator} of $w$ if $v$ dominates $w$ and $v \not= w$.
The dominator relation in $G$ can be represented by a tree rooted at $s$, the \emph{dominator tree} $D$, such that $v$ dominates $w$ if and only if $v$ is an ancestor of $w$ in $D$.
If $w \not= s$ is reachable, we denote by $d(w)$ the parent of $w$ in $D$.
Lengauer and Tarjan~\cite{domin:lt} presented an algorithm for computing dominators in  $O(m \alpha(m,n))$ time for a flow graph with $n$ vertices and $m$ edges, where $\alpha$ is a functional inverse of Ackermann's function~\cite{dsu:tarjan}.
Subsequently, several linear-time algorithms
were discovered~\cite{domin:ahlt,dominators:bgkrtw,dominators:Fraczak2013,Gabow:Poset:TALG}.
The dominator tree is a central tool in program optimization and code generation~\cite{cytron:91:toplas}, and it has applications in other diverse areas including constraint programming~\cite{QVDR:PADL:2006},
circuit testing~\cite{amyeen:01:vlsitest}, theoretical biology~\cite{foodwebs:ab04}, memory profiling~\cite{memory-leaks:mgr2010}, the analysis of diffusion networks~\cite{Rodrigues:icml12}, and in connectivity problems~\cite{2vc,2VCSS:Geo,2ECB,2VCB,GIN16:ICALP,2CC:HenzingerKL15,Italiano2012,2vcb:jaberi15,2VCC:Jaberi2016}.

A \emph{low-high order $\delta$ of $G$}~\cite{DomCert:TALG} is a preorder of the dominator tree $D$ such for all reachable vertices $v \not= s$, $(d(v), v) \in E$ or there are two edges $(u, v) \in E$, $(w,v) \in E$ such that $u$ and $w$ are reachable, $u$ is less than $v$
($u <_{\delta} v$), $v$ is less than $w$ ($v <_{\delta} w$), and $w$ is not a descendant of $v$ in $D$. See Figure \ref{figure:lowhigh}.
Every flow graph $G$ has a low-high order, computable in linear-time~\cite{DomCert:TALG}.
Low-high orders provide a correctness certificate for dominator trees that is straightforward to verify~\cite{domver:ZZ}.
By augmenting an algorithm that computes the dominator tree $D$ of a flow graph $G$ so that it also computes a low-high order of $G$,
one obtains a \emph{certifying algorithm} to compute $D$. (A \emph{certifying algorithm}~\cite{certifying} outputs both the solution and a correctness
certificate, with the property that it is straightforward to use the certificate to verify that the computed solution is correct.)
Low-high orders also have applications in path-determination problems~\cite{Tholey2012} and  in
fault-tolerant network design~\cite{FaultTolerantReachability,FaultTolerantReachability:STOC16,DomCert:TALG:Add}.

\begin{figure}[t!]
\begin{center}
\centerline{\includegraphics[trim={0 0 0 7cm}, clip=true, width=\textwidth]{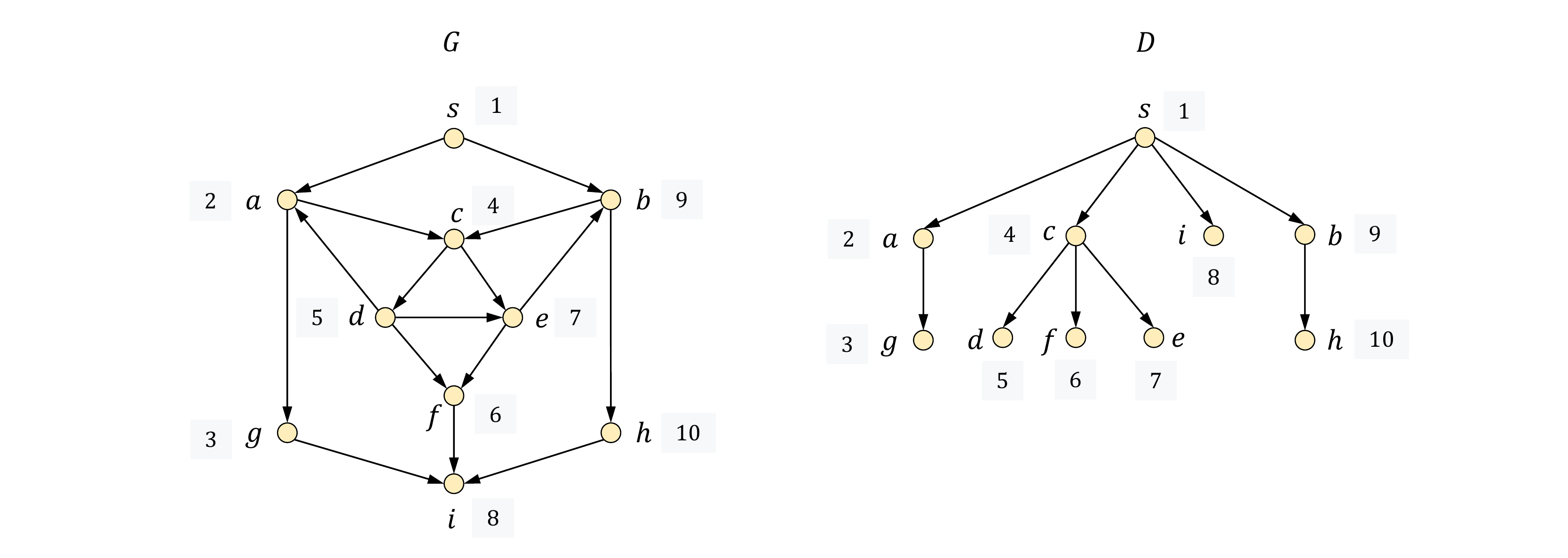}}
\centerline{\includegraphics[trim={0 0 0 7cm}, clip=true, width=\textwidth]{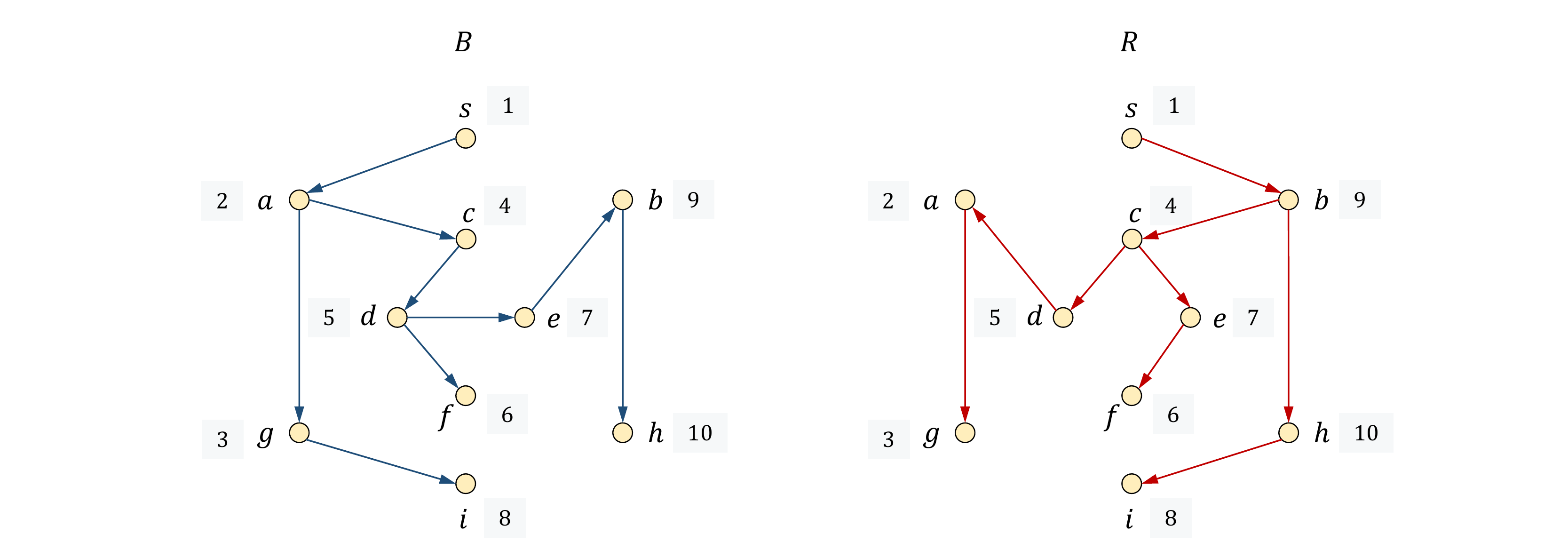}}
\caption{A flow graph $G$, its dominator tree $D$, and two strongly divergent spanning trees $B$ and $R$.
The numbers correspond to a preorder numbering of $D$ that is a low-high order of $G$.}
\label{figure:lowhigh}
\end{center}
\end{figure}

A notion closely related to low-high orders is that of divergent spanning trees~\cite{DomCert:TALG}.
Let $V_r$ be the set of reachable vertices, and let $G[V_r]$ be the flow graph with start vertex $s$ that is induced by $V_r$.
Two spanning trees $B$ and $R$ of $G[V_r]$, rooted at $s$, are \emph{divergent} if for all $v$, the paths from $s$ to $v$ in $B$ and $R$ share only the dominators of $v$; $B$ and $R$ are \emph{strongly divergent} if for every pair of vertices $v$ and $w$, either the path in $B$ from $s$ to $v$ and the path in $R$ from $s$ to $w$ share only the common dominators of $v$ and $w$, or the path in $R$ from $s$ to $v$ and the path in $B$ from $s$ to $w$ share only the common dominators of $v$ and $w$. 
In order to simplify our notation, we will refer to $B$ and $R$, with some abuse of terminology, as strongly divergent spanning trees of $G$.
Every flow graph has a pair of strongly divergent spanning trees.
Given a low-high order of $G$, it is straightforward to compute two strongly divergent spanning trees of $G$ in $O(m)$ time~\cite{DomCert:TALG}.
Divergent spanning trees can be used in data structures that compute pairs of vertex-disjoint $s$-$t$ paths in $2$-vertex connected digraphs (for any two query vertices $s$ and $t$)~\cite{2vc}, in fast algorithms for approximating the smallest $2$-vertex-connected spanning subgraph of a digraph~\cite{2VCSS:Geo}, and in constructing sparse subgraphs of a given digraph that maintain certain connectivity requirements \cite{2ECB,2vcb:jaberi15,2VCC:Jaberi2016}.

\begin{figure}[h!]
\begin{center}
\centerline{\includegraphics[trim={0 0 0 7cm}, clip=true, width=\textwidth]{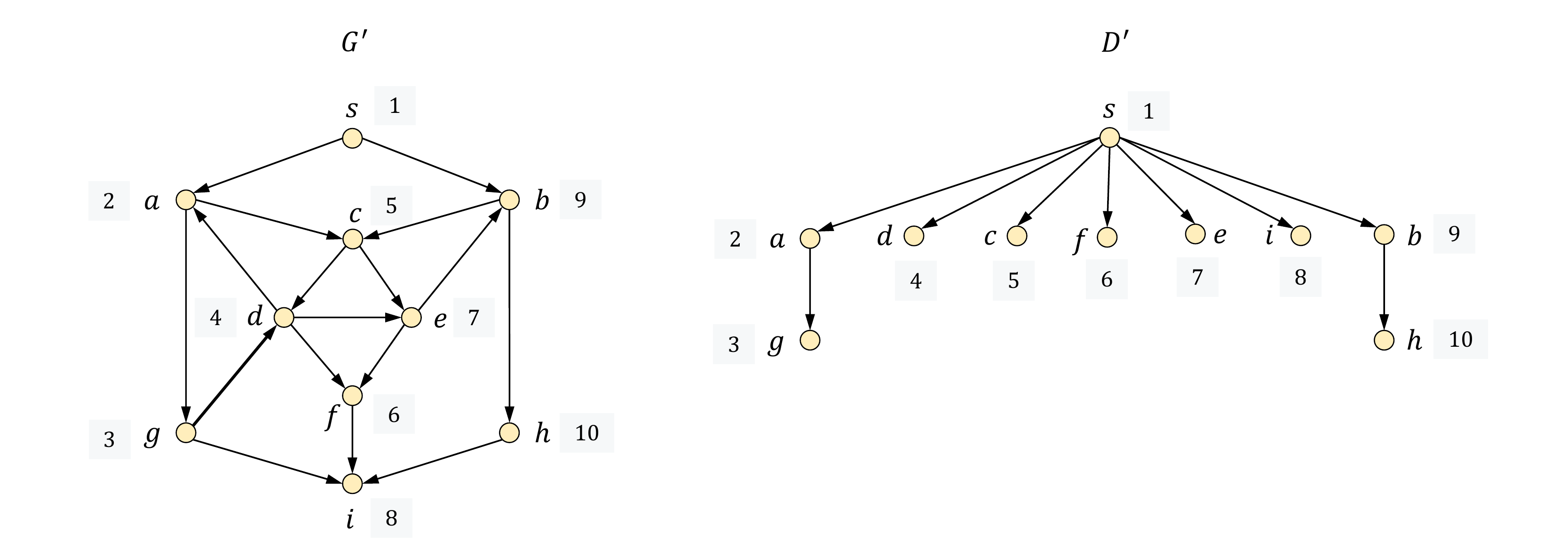}}
\centerline{\includegraphics[trim={0 0 0 7cm}, clip=true, width=\textwidth]{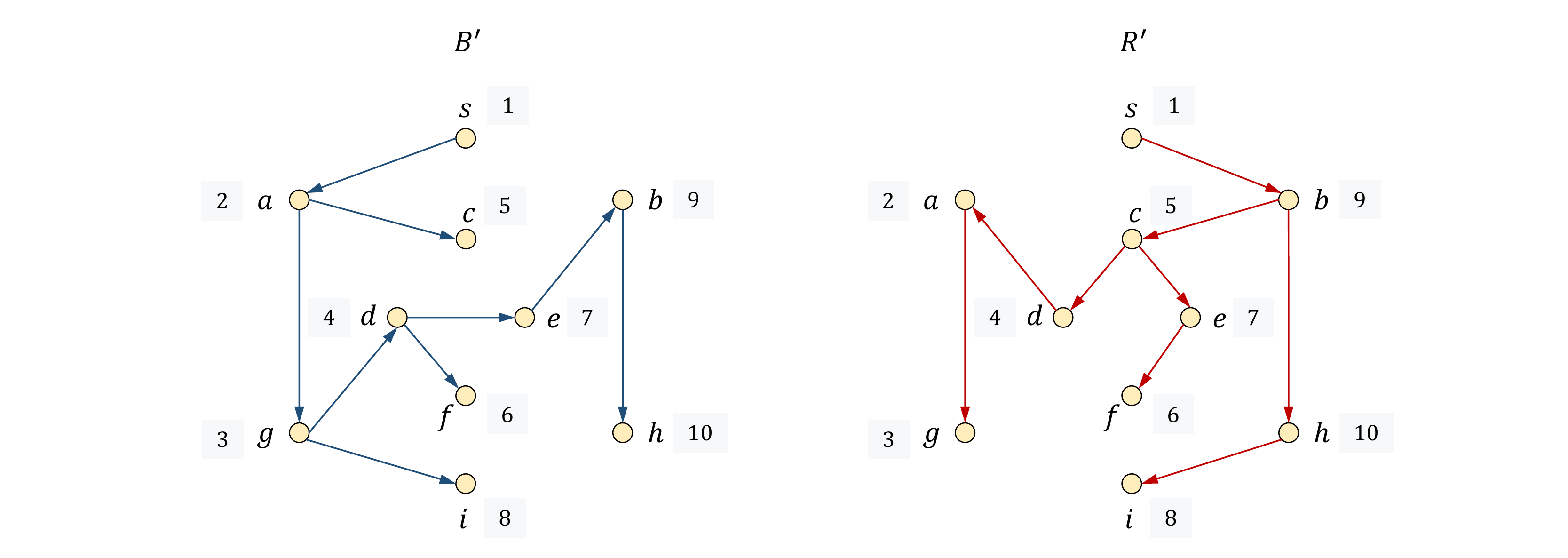}}
\caption{The flow graph of Figure \ref{figure:lowhigh} after the insertion of edge $(g,d)$, and its updated dominator tree $D'$ with a low-high order, and two strongly divergent spanning trees $B'$ and $R'$.}
\label{figure:lowhigh2}
\end{center}
\end{figure}

In this paper we consider how to update a low-high order of a flow graph through a sequence of
edge insertions. See Figure \ref{figure:lowhigh2}.
The difficulty in updating the dominator tree and a low-high order is due to the following facts.
An affected vertex can be arbitrarily far from the inserted edge, and a single edge insertion may cause $O(n)$ parent changes in $D$.
Furthermore, since a low-high order is a preorder of $D$, a single edge insertion may cause $O(n)$ changes in this order,
even if there is only one vertex that is assigned a new parent in $D$ after the insertion.
More generally, we note that
the hardness of dynamic algorithms on digraphs has been recently supported also by conditional lower bounds~\cite{AW14}.
Our first contribution is to show that we can maintain a low-high order of a flow graph $G$ with $n$ vertices through a sequence of
$m$ edge insertions in $O(mn)$ total time. Hence, we obtain a substantial improvement over the naive solution of recomputing
a low-high order from scratch after each edge insertion, which takes $O(m^2)$ total time.
Our result also implies the first incremental \emph{certifying algorithms}~\cite{certifying} for computing dominators in $O(mn)$ total time, which
answers an open question in \cite{DomCert:TALG}.
We present two algorithms that achieve this bound, a simple and a more sophisticated.
Both algorithms combine the incremental dominators algorithm of \cite{dyndom:2012} with the linear-time computation of two divergent spanning trees
from \cite{DomCert:TALG}. Our more sophisticated algorithm also applies a slightly modified version of a static low-high algorithm from
\cite{DomCert:TALG} on an auxiliary graph.
Although both algorithms have the same worst-case running time, our experimental results show that the sophisticated algorithm is by far superior in practical scenarios.
%
%
\ignore{
\begin{figure}[t!]
\begin{center}

\begin{tikzpicture}
    \node[anchor=south west,inner sep=0] (image) at (0,0) {\includegraphics[width=\textwidth]{figure2crop.pdf}};
    \begin{scope}[x={(image.south east)},y={(image.north west)}]

\node at (0.05,0.90) {$G$};
\node at (0.32,0.90) {$D$};
\node at (0.57,0.90) {$B$};
\node at (0.83,0.90) {$R$};
\end{scope}
\end{tikzpicture}
\caption{A flow graph $G$, its dominator tree $D$, and two strongly divergent spanning trees $B$ and $R$. The numbers correspond to a preorder numbering of $D$ that is a low-high order of $G$.}
\label{figure:lowhigh}
\end{center}
\vspace{-.5cm}
\end{figure}
}

We note that the incremental dominators problem arises in various applications, such as incremental data flow analysis and compilation~\cite{semidynamic-digraphs:CFNP,Gargi:2002,RR:incdom,SGL}, distributed authorization~\cite{Mowbray:2008},
and in incremental algorithms for maintaining $2$-connectivity relations in directed graphs~\cite{GIN16:ICALP}.
In Section \ref{sec:applications} we show how our result on incremental low-high order maintenance implies the following incremental algorithms that also run in $O(mn)$ total time for a sequence
of $m$ edge insertions.
\begin{itemize}
\item First, we give an algorithm that maintains,
after each edge insertion, two strongly divergent spanning trees of $G$, and answers the following queries
in constant time: (i) For any two query vertices $v$ and $w$, find a path $\pi_{sv}$ from $s$ to $v$ and a path $\pi_{sw}$ from
$s$ to $w$, such that $\pi_{sv}$ and $\pi_{sw}$ share only the common dominators of $v$ and $w$. We can output
these paths in $O(|\pi_{sv}| + |\pi_{sw}|)$ time. (ii) For any two query vertices $v$ and $w$, find a path $\pi_{sv}$ from $s$ to $v$ that avoids
$w$, if such a path exists. We can output this path in $O(|\pi_{sv}|)$ time.

\item Then we provide an algorithm for an incremental version of the fault-tolerant reachability problem~\cite{FaultTolerantReachability,FaultTolerantReachability:STOC16}.
We maintain a flow graph $G=(V,E,s)$ with $n$ vertices through a sequence
of $m$ edge insertions, so that we can answer the following query in $O(n)$ time.
Given a spanning forest $F =(V, E_F)$ of $G$ rooted at $s$, find a set of edges $E' \subseteq E \setminus E_F$ of
minimum cardinality, such that the subgraph $G' = (V, E_F \cup  E', s)$ of $G$ has the same dominators as $G$.

\item Finally, given a digraph $G$, we consider how to maintain incrementally a spanning subgraph of $G$ with $O(n)$ edges that
preserves the $2$-edge-connectivity relations in $G$.
\end{itemize}

We also revisit the problem of computing a smallest $2$-vertex-connected spanning subgraph (2VCSS)  of a directed graph~\cite{CT00,2VCSS:Geo}.
We present a linear-time algorithm that computes a $2$-approximation of the smallest $2$-vertex-connected spanning subgraph
(\textsf{2VCSS}) of a $2$-vertex-connected digraph. This improves significantly the best previous approximation ratio
achievable in linear time for this problem, which was $3$~\cite{2VCSS:Geo}. Our new algorithm is also shown to achieve better performance in practice.


\section{Preliminaries}
\label{sec:definitions}

Let $G=(V,E,s)$ be a flow graph with start vertex $s$, and let $D$ be the dominator tree of $G$.
A spanning tree $T$ of $G$ is a tree with root $s$ that contains a path from $s$ to $v$ for all reachable vertices $v$.
We refer to a spanning subgraph $F$ of $T$ as a spanning forest of $G$.
Given a rooted tree $T$, we denote by $T(v)$ the subtree of $T$ rooted at $v$ (we also view $T(v)$ as the set of descendants of $v$).
Let $T$ be a tree rooted at $s$ with vertex set $V_T \subseteq V$, and let $t(v)$ denote the parent of a vertex $v \in V_T$ in $T$.
If $v$ is an ancestor of $w$, $T[v, w]$ is the path in $T$ from $v$ to $w$.
In particular, $D[s,v]$ consists of the vertices that dominate $v$.
If $v$ is a proper ancestor of $w$, $T(v, w]$ is the path to $w$ from the child of $v$ that is an ancestor of $w$. Tree $T$ is \emph{flat} if its root is the parent of every other vertex.
Suppose now that the vertex set $V_T$ of $T$ consists of the vertices reachable from $s$.
Tree $T$ has the \emph{parent property} if for all $(v, w) \in E$ with $v$ and $w$ reachable, $v$ is a descendant of $t(w)$ in $T$.
If $T$ has the parent property and has a low-high order, then $T=D$~\cite{DomCert:TALG}.
For any vertex $v \in V$, we denote by $C(v)$ the set of children of $v$ in $D$.
A \emph{preorder} of $T$ is a total order of the vertices of $T$ such that, for every vertex $v$, the descendants of $v$ are ordered consecutively, with $v$ first.
Let $\zeta$ be a preorder of $D$.
Consider a vertex $v \not= s$. We say that \emph{$\zeta$ is a low-high order for $v$ in $G$}, if
$(d(v), v) \in E$ or there are two edges $(u, v) \in E$, $(w,v) \in E$ such that $u <_{\zeta} v$ and $v <_{\zeta} w$,
and $w$ is not a descendant of $v$ in $D$.
Given a graph $G=(V,E)$ and a set of edges $S \subseteq V \times V$, we denote by $G \cup S$ the graph obtained by inserting into $G$ the edges of $S$.

\section{Incremental low-high order}
\label{sec:incremental}

In this section we describe two algorithms to maintain a low-high order of a digraph through a sequence of edge insertions.
We first review some useful facts for updating a dominator tree after an edge insertion \cite{dynamicdominator:AL,dyndom:2012,RR:incdom}.
Let $(x,y)$ be the edge to be inserted. We consider the effect of this insertion when both $x$ and $y$ are reachable.
Let $G'$
be the flow graph that results from $G$ after inserting $(x,y)$.
Similarly, if $D$ is the dominator tree of $G$ before the insertion, we let $D'$ be the the dominator tree of $G'$.
Also, for any function $f$ on $V$, we let $f'$ be the function after the update.
We say that vertex $v$ is \emph{affected} by the update if $d(v)$ (its parent in $D$) changes, i.e.,
$d'(v) \not= d(v)$. We let $A$ denote the set of affected vertices.
Note that we can have 
$D'[s,v] \not= D[s,v]$ even if $v$ is not affected.
We let
$\mathit{nca}(x,y)$ denote the nearest common ancestor of $x$ and $y$ in the dominator tree $D$.
We also denote by $\mathit{depth}(v)$ the depth of a reachable vertex $v$ in $D$.
There are affected vertices after the insertion of $(x,y)$ if and only if $\mathit{nca}(x,y)$ is not a descendant of $d(y)$~\cite{irdom:rr94}.
A characterization of the affected vertices is provided by the following lemma, which is a refinement of a result in~\cite{dynamicdominator:AL}.

\begin{lemma}
\label{lemma:insert-affected} \emph{(\cite{dyndom:2012})}
Suppose $x$ and $y$ are reachable vertices in $G$.
A vertex $v$ is affected after the insertion of edge $(x,y)$ if and only if
$ \mathit{depth}(\mathit{nca}(x,y)) < \mathit{depth}(d(v))$
and there is a path $\pi$ in $G$
from $y$ to $v$ such that $\mathit{depth}(d(v)) < \mathit{depth}(w)$ for all $w \in \pi$.
If $v$ is affected, then it becomes a child of $\mathit{nca}(x,y)$ in $D'$, i.e., $d'(v)=\mathit{nca}(x,y)$.
\end{lemma}

The algorithm (\textsf{DBS}) in \cite{dyndom:2012} applies Lemma \ref{lemma:insert-affected} to identify
affected vertices by starting a search from $y$ (if $y$ is not affected, then no other vertex is).
To do this search for affected vertices,
it suffices to maintain
the outgoing and incoming edges of each vertex.
These sets are organized as singly
linked lists, so that a new edge can be inserted in $O(1)$ time.
The dominator tree $D$
is represented by the parent function $d$. We also maintain the
depth in $D$ of each reachable vertex.
We say that a vertex $v$
is \emph{scanned}, if the edges leaving $v$ are examined during the search for affected vertices, and that it is \emph{visited} if there is a scanned vertex $u$ such
that $(u,v)$ is an edge in $G$.
By Lemma \ref{lemma:insert-affected}, a visited vertex $v$ is scanned if $\mathit{depth}(\mathit{nca}(x,y)) < \mathit{depth}(d(v))$.

\begin{lemma}\emph{(\cite{dyndom:2012})}
\label{lemma:scanned-vertices}
Let $v$ be a scanned vertex. Then $v$ is a descendant of an affected vertex in $D$.
\end{lemma}

\subsection{Simple Algorithm}
\label{sec:incremental-simple}

In this algorithm we maintain, after each insertion, a subgraph $H=(V, E_H)$ of $G$ with $O(n)$ edges that has the same dominator tree as $G$.
Then, we can compute a low-high order $\delta$ of $H$ in $O(|E_H|)=O(n)$ time. Note that $\delta$ is also a valid low-high order of $G$.
Subgraph $H$ is formed by the edges of two divergent spanning trees $B$ and $R$ of $G$.
After the insertion of an edge $(x,y)$, where both $x$ and $y$ are reachable,
we form a graph $H'$ by inserting into $H$ a set of edges $\mathit{Last}(A)$ found during the search for affected vertices.
Specifically, $\mathit{Last}(A)$ contains edge $(x,y)$ and, for each affected vertex $v \not= y$, the last edge on a path $\pi_{yv}$ that satisfies
Lemma \ref{lemma:insert-affected}.
Then, we set $H' = H \cup \mathit{Last}(A)$. Finally, we compute a low-high order and two divergent spanning trees of $H'$, which are also valid for $G'$.
Algorithm \textsf{SimpleInsertEdge} describes this process.

\begin{algorithm}
\LinesNumbered
\DontPrintSemicolon
\KwIn{Flow graph $G=(V,E,s)$, its dominator tree $D$, a low-high order $\delta$ of $G$, two divergent spanning trees $B$ and $R$ of $G$, and a new edge $e=(x,y)$.}
\KwOut{Flow graph $G' = (V, E \cup (x,y), s)$, its dominator tree $D'$, a low-high order $\delta'$ of $G'$, and two divergent spanning trees $B'$ and $R'$ of $G'$.}

Insert $e$ into $G$ to obtain $G'$.\;

\lIf{$x$ is unreachable in $G$}{\KwRet $(G', D,\delta, B, R)$}
\ElseIf{$y$ is unreachable in $G$}{
$(D', \delta', B' , R') \leftarrow \mathsf{ComputeLowHigh}(G')$\;
\KwRet $(G', D', \delta', B', R')$
}

Let $H = B \cup R$.\;

Compute the updated dominator tree $D'$ of $G'$ and return a list $A$ of the affected vertices, and a list $\mathit{Last}(A)$ of the last edge entering each $v \in A$ in a path of Lemma 3.2.\;

Compute the subgraph $H' = H  \cup \mathit{Last}(A)$ of $G'$.\;

Compute $(D', \delta', B' , R') \leftarrow \mathsf{ComputeLowHigh}(H')$\;

\KwRet $(G', D', \delta', B', R')$

\caption{\textsf{SimpleInsertEdge}$(G, D, \delta, B, R, e)$}
\end{algorithm}

Note that when only $x$ is reachable before the insertion, we re-initialize our algorithm by running the linear-time algorithm of
\cite[Section 6]{DomCert:TALG}, which returns both a low-high order and two divergent spanning trees.

\begin{lemma}
\label{lemma:SimpleCorrect}
Algorithm \textsf{SimpleInsertEdge} is correct.
\end{lemma}
\begin{proof}
It suffices to show that subgraph $H'$ of $G'$, computed in line 9, has the same dominator tree with $G'$.
Note that graph $H$, formed by two divergent spanning trees $B$ and $R$ of $G$ in line 7, has the same dominator tree $D$ as $G$.
Hence, since $\mathit{Last}(A)$ contains $(x,y)$, the immediate dominator of $y$ is the same in $H'$ and in $G'$.
Let $A$ be the set of affected vertices in $G$ after the insertion of edge $(x,y)$.
Since $H'$ is a subgraph of $G'$, any vertex in  $V \setminus A$ has the same immediate dominator in $H'$ and in $G'$.
It remains to argue that for each vertex $v \in A \setminus y$, there is a path $\widehat{\pi}_{yv}$ in $H'$ that satisfies Lemma \ref{lemma:insert-affected}.
Let $\pi_{yv}$ be the path from $y$ in $v$ in $G$ that was found by the search for affected vertices performed by the algorithm of \cite{dyndom:2012}.
We give a corresponding path $\widehat{\pi}_{yv}$ in $H'$. Recall that every vertex on $\pi_{yv}$ is scanned and that every scanned vertex is a descendant in $D$
of an affected vertex.
We argue that for every two successive affected vertices $u$ and $w$ on $\pi_{yv}$ there is a path $\pi_{uw}$ from $u$ to $w$ in $H'$ that consists of vertices of depth at least $\mathit{depth}(w)$.
Note that, by properties of the depth-based search,  $\mathit{depth}(w) \le \mathit{depth}(u)$.
Indeed, let $(p,w)$ be the edge entering $w$ from $\pi_{yv}$. Then $(p,w) \in \mathit{Last}(A)$ and $p$ is a descendant of $u$ in $D$.
Also, since $u$ dominates $p$ in $G$, $u$ is an ancestor of $p$ in both spanning trees $B$ and $R$. We let $\widehat{\pi}_{uw} = B[u,p] \cdot (p,w)$.
All vertices on $B[u,p]$ are dominated by $u$, since otherwise there would be a path from $s$ to $p$ avoiding $u$.
So, $\widehat{\pi}_{uw}$ is path from $u$ to $w$ in $H'$ that consists of vertices with depth at least $\mathit{depth}(w)$.
\end{proof}

\begin{lemma}
\label{lemma:SimpleRunningTime}
Algorithm \textsf{SimpleInsertEdge} maintains a low-high order of a flow graph $G$ with $n$ vertices through a sequence of edge insertions in
$O(mn)$ total time, where $m$ is the total number of edges in $G$ after all insertions.
\end{lemma}
\begin{proof}
Consider the insertion of an edge $(x,y)$. If $y$ was unreachable in $G$
then we compute $D$, two divergent spanning trees $B$ and $R$, and a low-high order in $O(m)$ time~\cite{DomCert:TALG}.
Throughout
the whole sequence of $m$ insertions, such an event can happen $O(n)$ times,
so all insertions to unreachable vertices are handled in $O(mn)$ total time.

Now we consider the cost of executing \textsf{SimpleInsertEdge}.
when both $x$ and $y$ are reachable in $G$.
Let $\nu$ be the number of scanned vertices, and let $\mu$ be the number of their adjacent edges.
We can update the dominator tree and locate the affected vertices (line 8) in $O(\nu+\mu+n)$ time~\cite{dyndom:2012}.
At the same time we can compute the edge set $\mathit{Last(A)}$.
Computing $H'$ in line 9 takes $O(n)$ time since $B \cup R \cup \mathit{Last(A)}$ contains at most $3(n-1)$ edges.
Also, computing the dominator tree, two divergent spanning trees,, and a low-high order of $H'$ in $O(n)$ time~\cite{DomCert:TALG}.
So \textsf{SimpleInsertEdge} runs in $O(\nu+\mu+n)$ time. .
The $O(n)$ term gives a total cost of $O(mn)$ for the whole sequence of $m$ insertions.
We distribute the remaining $O(\nu+\mu)$ cost to the scanned vertices and edges, that is $O(1)$ per scanned vertex or edge.
Since the depth in $D$ of every scanned vertex decreases by at least one, a vertex and an edge can be scanned at most
$O(n)$ times. Hence, each vertex and edge can contribute at most $O(n)$ total cost through the whole sequence of $m$ insertions.
The $O(m n)$ bound follows.
\end{proof}

\subsection{Efficient Algorithm}
\label{sec:efficient}

Here we develop a more practical algorithm that maintains a low-high order $\delta$ of a flow graph $G=(V,E,s)$ through a sequence of edge insertions.
Our algorithm uses the incremental dominators algorithm of \cite{dyndom:2012} to update the dominator tree $D$ of $G$
after each edge insertion.
We describe a process to update $\delta$ based on new results
on the relation among vertices in $D$ that are affected by the insertion. .
These results enable us to identify a subset of vertices for which we can compute a ``local'' low-high order, that
can be extended to a valid low-high order of $G$ after the update. We show that such a ``local'' low-high order can be computed
by a slightly modified version of an algorithm from \cite{DomCert:TALG}. We apply this algorithm on a sufficiently small
flow graph that is defined by the affected vertices, and is constructed using the concept of derived edges~\cite{path:tarjan81}.

\subsubsection{Derived edges and derived flow graphs}

Derived graphs, first defined in \cite{path:tarjan81}, reduce the problem of finding a low-high order to the case of a flat dominator tree~\cite{DomCert:TALG}.
By the parent property of $D$, if $(v, w)$ is an edge of $G$, the parent $d(w)$ of $w$ is an ancestor of $v$ in $D$.
Let $(v,w)$ be an edge of $G$, with $w$ not an ancestor of $v$ in $D$.
Then, the \emph{derived edge} of $(v, w)$ is the edge $(\overline{v}, w)$, where $\overline{v} = v$ if $v = d(w)$, $\overline{v}$ is the sibling of $w$ that is an ancestor of $v$ if $v \not= d(w)$. If $w$ is an ancestor of $v$ in $D$, then the derived edge of $(v, w)$ is null. Note that a derived edge $(\overline{v}, w)$ may not be an original edge of $G$. 
For any vertex $w \in V$ such that $C(w) \not= \emptyset$, we define the \emph{derived flow graph of $w$}, denoted by $G_w = (V_w, E_w, w)$, as the flow graph with
start vertex $w$,
vertex set $V_w = C(w) \cup \{ w \}$, and edge set $E_w = \{ (\overline{u}, v) \ | \ v \in V_w \mbox{ and } (\overline{u}, v) \mbox{ is the non-null derived edge of some edge in } E \}$.
By definition, $G_w$ has flat dominator tree, that is, $w$ is the only proper dominator of any vertex $v \in V_w \setminus w$.
We can compute a low-high order $\delta$ of $G$ by computing a low-high order $\delta_w$ in each derived flow graph
$G_w$. Given these low-high orders $\delta_w$, we can compute a low-high order of $G$ in $O(n)$ time by a depth-first traversal of $D$.
During this traversal, we visit the children of each vertex $w$ in their order in $\delta_w$, and number the vertices from $1$ to $n$ as they are visited.  The resulting preorder of $D$ is low-high on $G$.
Our incremental algorithm identifies, after each edge insertion, a specific derived flow graph $G_w$
for which a low-high order $\delta_w$ needs to be updated.
Then, it uses $\delta_w$ to update the low-high order of the whole flow graph $G$.
Still, computing a low-high order of $G_w$ can be too expensive to give us the desired running time.
Fortunately, we can overcome this obstacle by exploiting a relation among the vertices that are affected
by the insertion, as specified below. This allows us to compute $\delta_w$ in a contracted version of $G_w$.

\subsubsection{Affected vertices}

Let $(x,y)$ be the inserted vertex, where both $x$ and $y$ are reachable. Consider the execution of algorithm \textsf{DBS}~\cite{dyndom:2012} that updates the dominator tree by applying Lemma \ref{lemma:insert-affected}.
Suppose vertex $v$ is scanned, and let $q$ be the nearest affected ancestor of $v$ in $D$. Then, by Lemma \ref{lemma:insert-affected}, vertex $q$ is a child of $\mathit{nca}(x,y)$ in $D'$, i.e., $d'(q)=\mathit{nca}(x,y)$, and $v$ remains a descendant of $q$ in $D'$.

\begin{lemma}
\label{lemma:simple}
Let $u$ and $v$ be vertices such that $u \in D(v)$. Then, any simple path from $v$ to $u$ in $G$ contains only vertices in $D(v)$.
\end{lemma}
\begin{proof}
Since $u \in D(v)$, $v$ dominates $u$, so all paths from $s$ to $u$ contain $v$.
Let $\pi_{vu}$ be a simple path from $v$ to $u$. Suppose, for contradiction, that $\pi_{vu}$ contains a vertex
$w \not\in D(v)$. Let $\pi_{wu}$ be the part of $\pi_{vu}$ from $w$ to $u$. Since $w \not\in D(v)$,
there is a path $\pi_{sw}$ from $s$ to $w$ that avoids $v$. But then $\pi_{sw} \cdot \pi_{wu}$ is a path from
$s$ to $u$ that avoids $v$, a contradiction.
\end{proof}

\begin{lemma}
\label{lemma:also-affected}
Let $v$ be vertex that is affected by the insertion of $(x,y)$, and let $w$ be a sibling of $v$ in $D$. If there is an edge $(u,w)$ with $u$ a descendant of $v$ in $D$ then $w$ is also affected.
\end{lemma}
\begin{proof}
Since $v$ is affected, there is a path $\pi_{yv}$ from $y$ to $v$ in $G$ that satisfies Lemma \ref{lemma:insert-affected}.
By Lemma \ref{lemma:simple} and the fact that $u$ is a descendant of $v$ in $D$, there is a simple path $\pi_{vu}$ from $v$ to $u$ in $G$ that contains only vertices in $D(v)$.
Thus, $\pi_{yv} \cdot \pi_{vu} \cdot (u,w)$ is a path from $y$ to $w$ that also satisfies Lemma \ref{lemma:insert-affected}. Hence, $w$ is affected.
\end{proof}

\begin{lemma}
\label{lemma:parent}
Let $v$ be an ancestor of $w$ in $D$, and let $u$ be a vertex that is not a descendant of $v$ in $D$.
Then any path from $u$ to $w$ contains $v$.
\end{lemma}
\begin{proof}
Let $\pi_{uw}$ be a path from $u$ to $w$. Since $u$ is not a descendant of $v$, there is a path $\pi_{su}$
from $s$ to $u$ that avoids $v$. Hence, if $\pi_{uw}$ does not contain $v$, then $\pi_{su} \cdot \pi_{uw}$ is
path from $s$ to $w$ that avoids $v$, a contradiction.
\end{proof}

Our next lemma provides a key result about the relation of the affected vertices in $D$.

\begin{lemma}
\label{lemma:affected-child}
All vertices that are affected by the insertion of $(x,y)$ are descendants of a common child $c$ of $\mathit{nca}(x,y)$.
\end{lemma}
\begin{proof}
Let $z=\mathit{nca}(x,y)$, and
let $c$ be the child of $z$ that is an ancestor of $y$ in $D$. We claim that all affected vertices are descendants of $c$ in $D$. Suppose, for contradiction, that there is an affected vertex $v$ that is not a descendant of $c$ in $D$.
By Lemma \ref{lemma:insert-affected}, $v$ must be a descendant $z$ in $D$. Also, since
the children of $z$ are not affected, $v$ is not a child of $z$.
Hence, $v$ is a proper descendant of another child $q$ of $z$ in $D$ ($q \not= c$).
Let $\pi_{yv}$ be a path from $y$ to $v$ in $G$ that satisfies Lemma \ref{lemma:insert-affected}.
Since $y$ is not a descendant of $q$, by Lemma \ref{lemma:parent} path $\pi_{yv}$ must contain $q$.
But then $\pi_{yv}$ contains a vertex of depth $\mathit{depth}(d(v))$ or less, which contradicts Lemma \ref{lemma:insert-affected}.
\end{proof}

We shall apply Lemma \ref{lemma:affected-child} to construct a flow graph $G_A$ for the affected vertices.
Then, we shall use $G_A$ to compute a ``local'' low-high order that we extend to a valid low-high order of $G'$.

\subsubsection{Low-high order augmentation}

Let $\delta$ be a low-high order of $G$, and let $\delta'$ be a preorder of the dominator tree $D'$ of $G'$.
We say that $\delta'$ \emph{agrees with} $\delta$ if the following condition holds for any pair of siblings
$u, v$ in $D$ that are not affected by the insertion of $(x,y)$: $u <_{\delta'} v$ if and only if $u <_{\delta} v$.
Our goal is to show that there is a low-high order $\delta'$ of $G'$ that agrees with $\delta$.

\begin{lemma}
\label{lemma:agree}
Let $\delta$ be a low-high order of $G$ before the insertion of $(x,y)$.
There is a preorder $\delta'$ of $D'$ that agrees with $\delta$.
\end{lemma}
\begin{proof}
By Lemma \ref{lemma:insert-affected}, all affected vertices become children of $z$ in $D'$.
Hence, $C'(z) \supseteq C(z)$, and for any $v \not= z$, $C'(v) \subseteq C(v)$.
Then, for each vertex $v$, we can order the children of $v$ in $C'(v)$ that are not affected according to $\delta$.
Finally, we insert the affected vertices in any order in the list of children of $z$.
Let $\delta'$ be the preorder of $D'$ that is constructed by a depth-first traversal of $D'$
that visits the children of each vertex $w$ in the order specified above.
Then, $\delta'$ agrees with $\delta$.
\end{proof}

\begin{lemma}
\label{lemma:unaffected}
Let $\delta'$ be a preorder of $D'$ that agrees with $\delta$.
Let $v$ be a vertex that is not a child of $\mathit{nca}(x,y)$ and is not affected by the insertion of $(x,y)$.
Then $\delta'$ is a low-high order for $v$ in $G'$.
\end{lemma}
\begin{proof}
Since $v$ is not affected, $d(v)$ is still the parent of $v$ in $D$ after the insertion. So, if $(d(v),v) \in E$,
then $\delta'$ is a low-high order for $v$ in $G'$.
Now suppose that $(d(v),v) \not\in E$. Then there are two edges $(u,v)$ and $(w,v)$ in $E$ such that $u <_{\delta} v <_{\delta} w$,
where $w$ is not a descendant of $v$ in $D$.
Let $(\overline{u},v)$ and $(\overline{w},v)$ be the derived edges of $(u,v)$ and $(w,v)$, respectively, in $D$.
Then $\overline{u}$ and $\overline{w}$ are siblings of $v$ in $D$.
Siblings $\overline{u}$ and $\overline{w}$ exist and are distinct by the fact that $(d(v),v) \not\in E$ and by the parent property of $D$.
Hence, $\overline{u}  <_{\delta} v <_{\delta} \overline{w}$.
We argue that after the insertion of $(x,y)$, $\overline{u}$ (resp., $\overline{w}$) remains a sibling of $v$, and an ancestor of $u$ (resp., $w$).
If this is not the case, then there is an affected vertex $q$ on $D[\overline{u},u]$.
But then, Lemma \ref{lemma:also-affected} implies that $v$ is also be affected, a contradiction.
So, both $\overline{u}$ and $\overline{w}$ remain siblings of $v$ in $D'$, and $(\overline{u},v)$ and $(\overline{w},v)$
remain the derived edges of $(u,v)$ and $(w,v)$, respectively, in $D'$. Then, since $\delta'$ agrees with $\delta$,
$\delta'$ is a low-high order for $v$ in $G'$.
\end{proof}

We shall use Lemmata \ref{lemma:insert-affected} and \ref{lemma:unaffected} to show that in order to compute a low-high order of $G'$, it suffices to
compute a low-high order for the derived flow graph $G'_z$, where $z = \mathit{nca}(x,y)$.
Still, the computation of a low-high order of $G'_z$ is
too expensive to give us the desired running time. Fortunately, as we show next,
we can limit these computations for a contracted version of $G'_z$, defined by the affected vertices.

Let $\delta$ be a low-high order of $G$ before the insertion of $(x,y)$.
Also, let $z = \mathit{nca}(x,y)$, and let $\delta_z$ be
a corresponding low-high order of the derived flow graph $G_z$. That is, $\delta_z$ is the restriction of $\delta$
to $z$ and its children in $D$.
Consider the child $c$ of $z$ that, by Lemma \ref{lemma:affected-child}, is an ancestor of all the affected vertices.
Let $\alpha$ and $\beta$, respectively, be the predecessor and successor of $c$ in $\delta_z$. Note that $\alpha$ or $\beta$ may be null.
An \emph{augmentation of $\delta_z$} is an order $\delta'_z$ of $C'(z) \cup \{z\}$ that results from $\delta_z$
by inserting the affected vertices arbitrarily around $c$,
that is, each affected vertex is placed in an arbitrary position between $\alpha$ and $c$ or between $c$ and $\beta$.

\begin{lemma}
\label{lemma:unaffected-children}
Let $z = \mathit{nca}(x,y)$, and let $\delta_z$ be a low-high order of the derived flow graph $G_z$ before the insertion of $(x,y)$.
Also, let $\delta'_z$ be an augmentation of $\delta_z$, and let $\delta'$ be a preorder of $D'$ that extends $\delta'_z$.
Then, for each child $v$ of $z$ in $D$, $\delta'$ is a low-high order for $v$ in $G'$.
\end{lemma}
\begin{proof}
Since $v$ is a child of $z$ in $D$ it is not affected. Hence, $d'(v)=d(v)=z$.
Let $G'_z$ be the derived flow graph of $z$ after the insertion of $(x,y)$.
It suffices to show that $\delta'_z$ is a low-high order for $v$ in $G'_z$.

If $(z,v) \in E$, then $(z,v)$ is an edge in $G'_z$. So, in this case,
$\delta'_z$ is a low-high order for $v$ in $G'_z$.
Now suppose that $(z,v) \not\in E$.
Let $\delta$ be a preorder of $D$ that extends $\delta_z$.
Then, there are two edges $(u,v)$ and $(w,v)$ in $G$ such that $u <_{\delta} v <_{\delta} w$,
where $w$ is not a descendant of $v$ in $D$.
The fact that $(z,v)$ is not an edge implies that $u \not=z$ and $w \not= z$. 
Let $\overline{u}'$ (resp., $\overline{w}'$) be the nearest ancestor of $u$ (resp., $w$) in $D'$ that is a child of $z$.
We argue that $\overline{u}'$ exists and satisfies $\overline{u}' <_{\delta'_z} v$.
Let $\overline{u}$
be the nearest ancestor of $u$ in $D$.
If no vertex on $D(z,u]$ is affected, then $\overline{u}' = \overline{u}$.
Also, since $\overline{u} <_{\delta_z} v$ and by the fact that $\delta'_z$ is an augmentation of $\delta_z$, we have $\overline{u}' <_{\delta'_z} v$.
Suppose now that there is an affected vertex $q$ on $D(z,u]$. By Lemma \ref{lemma:insert-affected}, $q$ becomes
a child of $z$ in $D'$, hence $\overline{u}'= q$. Also, by Lemma \ref{lemma:affected-child},
$q$ is a proper descendant of $c$, so $\overline{u}=c$. Then $c <_{\delta_z} v$, and by the construction
of $\delta'_z$ we have $\overline{u}' <_{\delta'_z} v$.

An analogous argument shows that $\overline{w}'$ exists and satisfies $v <_{\delta'_z} \overline{w}'$.
Thus, $\delta'_z$ is a low-high order for $v$ in $G'_z$.
\end{proof}

\subsubsection{Algorithm}
\label{sec:algorithm}

Now we are ready to describe our incremental algorithm for maintaining a low-high order $\delta$ of $G$.
For each vertex $v$ that is not a leaf in $D$, we maintain a list of its children $C(v)$ in $D$, ordered by $\delta$.
Also, for each vertex $v\not= s$, we keep two variables $\mathit{low}(v)$ and $\mathit{high}(v)$. Variable $\mathit{low}(v)$ stores an
edge $(u,v)$ such that $u \not= d(v)$ and $u <_{\delta} v$; $\mathit{low}(v) = \mathit{null}$ if no such edge exists.
Similarly,
$\mathit{high}(v)$ stores an
edge $(w,v)$ such that and $v <_{\delta} w$ and $w$ is not a descendant of $v$ in $D$; $\mathit{high}(v) = \mathit{null}$ if no such edge exists.
These variables are useful in the applications that we mention in Section \ref{sec:applications}.
Finally, we mark each vertex $v$ such that $(d(v),v) \in E$.
For simplicity, we assume that the vertices of $G$ are numbered from $1$ to $n$, so we can store the above information in corresponding
arrays $\mathit{low}$, $\mathit{high}$, and $\mathit{mark}$.
Note that for a reachable vertex $v$,
we can have $\mathit{low}(v)=\mathit{null}$ or $\mathit{high}(v)=\mathit{null}$ (or both) only if
$\mathit{mark}(v) = \mathit{true}$.
Before any edge insertion, all vertices are unmarked, and all entries in arrays $\mathit{low}$ and $\mathit{high}$ are null.
We initialize the algorithm and the associated data structures by executing a linear-time algorithm to
compute the dominator tree $D$ of $G$~\cite{domin:ahlt,dominators:bgkrtw} and a linear-time algorithm
to compute a low-high order $\delta$ of $G$~\cite{DomCert:TALG}.
So, the initialization takes $O(m+n)$ time for a digraph with $n$ vertices and $m$ edges. Next, we describe the main routine to handle an edge insertion.
We let $(x,y)$ be the inserted edge.
Also, if $x$ and $y$ are reachable before the insertion, we let $z = \mathit{nca}(x,y)$.

\begin{algorithm}[h!]

\LinesNumbered
\DontPrintSemicolon

Compute the dominator tree $D$ and a low-high order $\delta$ of $G$.\;

\ForEach{reachable vertex $v \in V\setminus s$}{
\lIf{$(d(v),v) \in E$}{set $\mathit{mark}(v) \leftarrow \mathit{true}$}
find edges $(u,v)$ and $(w,v)$ such that $u <_{\delta} v <_{\delta} w$ and $w \not\in D(v)$\;
set $\mathit{low}(v) \leftarrow u$ and $\mathit{high}(v) \leftarrow w$}

\KwRet $(D, \delta, \mathit{mark}, \mathit{low}, \mathit{high})$

\caption{\textsf{Initialize}$(G)$}
\end{algorithm}

\begin{algorithm}[t]

\LinesNumbered
\DontPrintSemicolon
\KwIn{Flow graph $G=(V,E,s)$, its dominator tree $D$, a low-high order $\delta$ of $G$, arrays $\mathit{mark}$, $\mathit{low}$ and $\mathit{high}$, and a new edge $e=(x,y)$.}
\KwOut{Flow graph $G' = (V, E \cup (x,y), s)$, its dominator tree $D'$, a low-high order $\delta'$ of $G'$, and arrays $\mathit{mark}'$, $\mathit{low}'$ and $\mathit{high}'$.}

Insert $e$ into $G$ to obtain $G'$.\;

\lIf{$x$ is unreachable in $G$}{\KwRet $(G', D,\delta, \mathit{mark}, \mathit{low}, \mathit{high})$}
\ElseIf{$y$ is unreachable in $G$}{
$(D', \delta', \mathit{mark}', \mathit{low}', \mathit{high}') \leftarrow \mathsf{Initialize}(G')$\;
\KwRet $(G', D', \delta', \mathit{mark}', \mathit{low}', \mathit{high}')$
}

Compute the nearest common ancestor $z$ of $x$ and $y$ in $D$.\;

Compute the updated dominator tree $D'$ of $G'$ and return a list $A$ of the affected vertices.\;

\lForEach{vertex $v \in A$}{$\mathit{mark}'(y) \leftarrow \mathit{false}$}
\lIf{$z=x$}{$\mathit{mark}'(y) \leftarrow \mathit{true}$}

Execute \textsf{DerivedLowHigh$(z, A, \mathit{mark}')$}.\;

Make a dfs traversal of the subtrees of $D'$ rooted at each vertex $v \in A \cup \{c\}$ to compute $\delta'$. \;

\ForEach{vertex $v \in A \cup \{c\}$}{
find edges $(u,v)$ and $(w,v)$ such that $u <_{\delta'} v <_{\delta'} w$ and $w \not\in D'(v)$\;
set $\mathit{low}'(v) \leftarrow u$ and $\mathit{high}'(v) \leftarrow w$}

\KwRet $(G', D', \delta', \mathit{mark}', \mathit{low}', \mathit{high}')$

\caption{\textsf{InsertEdge}$(G, D, \delta, \mathit{mark}, \mathit{low}, \mathit{high}, e)$}

\end{algorithm}

Our main task now is to order the affected vertices according to a low-high order of $D'$.
To do this, we use an auxiliary flow graph $G_A=(V_A,E_A,z)$, with start vertex $z$, which
we refer to as the \emph{derived affected flow graph}.
Flow graph $G_A$ is essentially a contracted version of the derived flow graph $G'_z$ (i.e., the derived graph of $z$ after the insertion) as we explain later.
The vertices of the derived affected flow graph $G_A$ are $z$, the affected vertices
of $G$, their common ancestor $c$ in $D$ that is a child of $z$ (from Lemma \ref{lemma:affected-child}), and two auxiliary vertices $\alpha^{\ast}$ and $\beta^{\ast}$.
Vertex $\alpha^{\ast}$ (resp., $\beta^{\ast}$) represents vertices in $C(z)$ with lower (resp., higher) order in $\delta$ than $c$.
We include in $G_A$ the edges $(z,\alpha^{\ast})$ and $(z,\beta^{\ast})$. If $c$ is marked then we include the edge $(z,c)$ into $G_A$,
otherwise we add the edges $(\alpha^{\ast}, c)$ and $(\beta^{\ast},c)$ into $G_A$. Also, for each edge $(u,c)$ such that $u$
is a descendant of an affected vertex $v$, we add in $G_A$ the edge $(v,c)$.
Now we
specify the edges that enter an affected vertex $w$ in $G_A$.
We consider each edge $(u,w) \in E$ entering $w$ in $G$.
We have the following cases:
\begin{itemize}
\item[(a)] If $u$ is a descendant of an affected vertex $v$, we add in $G_A$ the edge $(v,w)$.
\item[(b)] If $u$ is a descendant of $c$ but not a descendant of an affected vertex, then we add in $G_A$ the edge $(c,w)$.
\item[(c)] If $u \not= z$ is not a descendant of $c$, then we add the edge $(\alpha^{\ast},w)$ if $u <_{\delta} c$, or the edge $(\beta^{\ast},w)$ if $c <_{\delta} u$.
\item[(d)] Finally, if $u = z$, then we add the edge $(z,w)$. (In cases (c) and (d), $u=x$ and $w=y$.)
\end{itemize}
See Figure \ref{figure:derived}.
Recall that $\alpha$ (resp., $\beta$) is the siblings of $c$ in $D$ immediately before (resp., after) $c$ in $\delta$, if it exists.
Then, we can obtain $G_A$ from $G'_z$ by contracting
all vertices $v$ with $v <_{\delta} c$ into $\alpha=\alpha^{\ast}$, and all vertices $v$ with $c <_{\delta} v$ into $\beta=\beta^{\ast}$.

\begin{figure}[t!]
\begin{center}
\centerline{\includegraphics[trim={0 0 0 11cm}, clip=true, width=\textwidth]{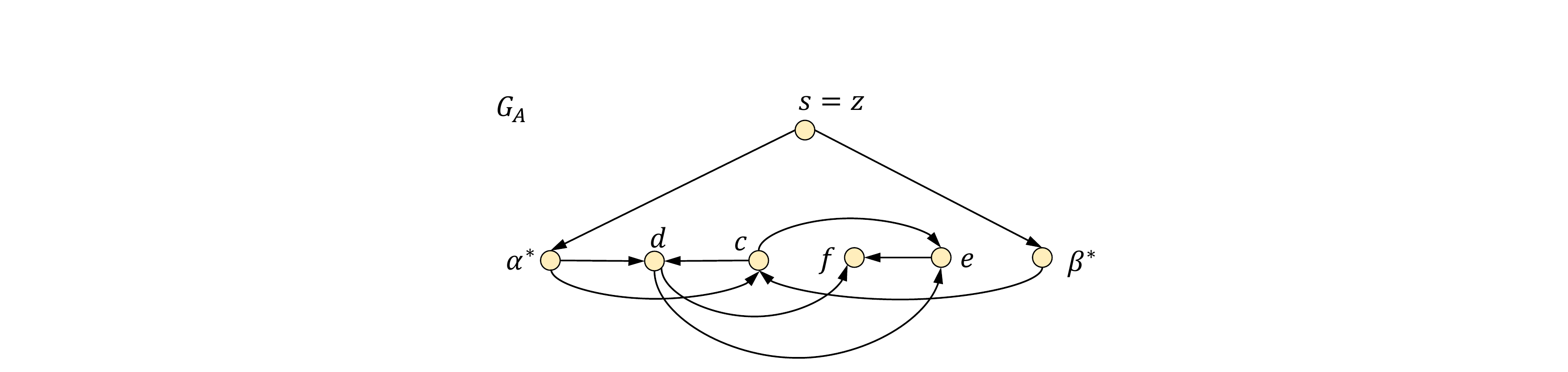}}
\caption{The derived affected flow graph $G_A$ that corresponds to the flow graph of Figure \ref{figure:lowhigh} after the insertion of edge $(g,d)$.}
\label{figure:derived}
\end{center}
\end{figure}

\begin{lemma}
\label{lemma:derived-affected-graph-flat}
The derived affected flow graph $G_A=(V_A, E_A, z)$ has flat dominator tree.
\end{lemma}
\begin{proof}
We claim that for any two distinct vertices $v, w \in V_A \setminus z$, $v$ does not dominate $w$. The lemma follows immediately from this claim.
The claim is obvious for $w \in \{\alpha^\ast, \beta^\ast \}$, since
$G_A$ contains the edges $(z,\alpha^\ast)$ and $(z,\beta^\ast)$.
The same holds for $w = c$, since $G_A$ contains the edge $(z,c)$, or both the edges $(\alpha^\ast, c)$ and $(\beta^\ast, c)$.
Finally, suppose $w \in V_A \setminus \{z, \alpha^{\ast}, \beta^{\ast}\}$. Then, by the construction of $G_A$, vertex $w$ is affected.
By Lemma \ref{lemma:affected-child}, $w \in D(c)$, so Lemma \ref{lemma:simple} implies that there is a path in $G$ from
$c$ to $w$ that contains only vertices in $D(c)$.
Hence, by construction, $G_A$ contains a path from $c$ to $w$ that avoids $\alpha^{\ast}$ and $\beta^{\ast}$, so
$\alpha^{\ast}$ and $\beta^{\ast}$ do not dominate $w$.
It remains to show that $w$ is not dominated in $G_A$ by $c$ or another affected vertex $v$.
Let $(x,y)$ be the inserted edge. Without loss of generality, assume that $c <_{\delta} x$.
Since $w$ is affected, there is a path $\pi$ in $G$ from $y$ to $w$
that satisfies Lemma \ref{lemma:insert-affected}. Then $\pi$ does not contain any vertex in $D[c,d(w)]$.
Also, by the construction of $G_A$, $\pi$ corresponds to a path $\pi_A$ in $G_A$ from $\beta^{\ast}$ to $y$ that avoids
any vertex in $A \cap D[c,d(w)]$.  Hence, $w$ is not dominated by any vertex in $A \cap D[c,d(w)]$. It remains to show that
$w$ is not dominated by any affected vertex $v$ in $A \setminus D[c,d(w)]$. Since both $v$ and $w$ are in $D(c)$ and
$v$ is not an ancestor of $w$ in $D$, there is a path $\pi'$ in $G$ from $c$ to $w$ that avoids $v$. By Lemma \ref{lemma:simple},
$\pi'$ contains only vertices in $D(c)$. Then, by the construction of $G_A$, $\pi'$ corresponds to a path $\pi'_A$ in $G_A$ from $c$ to $w$ that avoids
$v$. Thus, $v$ does not dominate $w$ in $G_A$.
\end{proof}

\begin{lemma}
\label{lemma:derived-affected-graph}
Let $\nu$ and $\mu$, respectively, be the number of scanned vertices and their adjacent edges.
Then, the derived affected flow graph $G_A$ has $\nu+4$ vertices, at most $\mu+5$ edges, and can be
constructed in $O(\nu+\mu)$ time.
\end{lemma}
\begin{proof}
The bound on the number of vertices and edges in $G_A$ follows from the definition of the derived affected flow graph.
Next, we consider the construction time of $G_A$.
Consider the edges entering the affected vertices. Let $w$ be an affected vertex, and let $(u,w) \not= (x,y)$ be
an edge of $G'$.
Let $q$ be nearest ancestor $u$ in $C'(z)$.
We distinguish two cases:
\begin{itemize}
\item $u$ is not scanned. In this case, we argue that $q=c$. Indeed, it follows from the parent property of $D$ and
Lemma \ref{lemma:affected-child} that both $u$ and $w$ are descendants of $c$ in $D$. Since $u$ is not scanned, no ancestor of
$u$ in $D$ is affected, so $u$ remains a descendant of $c$ in $D'$. Thus, $q=c$.
\item $u$ is scanned. Then, by Lemma \ref{lemma:scanned-vertices}, $q$ is the nearest affected ancestor of $u$ in $D$.
\end{itemize}
So we can construct the edges entering the affected vertices in $G_A$ in two phases. In the first phase we traverse the descendants of each affected vertex $q$ in $D'$.
At each descendant $u$ of $q$, we examine the edges leaving $u$.  When we find an edge $(u,w)$ with $w$ affected, then we insert into $G_A$ the edge
$(q,w)$. In the second phase we examine the edges entering each affected vertex $w$. When we find an edge $(u,w)$ with $u$ not visited during the first phase (i.e., $u$ was not scanned during the update of $D$), we insert into $G_A$ the edge $(c,w)$. Note that during this construction we may insert the same edge multiple times, but this does not affect the correctness or running time of our overall algorithm.
Since the descendants of an affected vertex are scanned, it follows that each phase runs in $O(\nu + \mu)$ time.

Finally, we need to consider the inserted edge $(x,y)$. Let $f$ be the nearest ancestor of $x$ that is in $C(z)$.
Since $y$ is affected, $c \not= f$. Hence, we insert into $G_A$ the edge $(\beta^{\ast},y)$ if $c <_{\delta} f$,
and the edge $(\alpha^{\ast},y)$ if $f <_{\delta} c$. Note that $f$ is found during the computation of $z=\mathit{nca}(x,y)$,
so this test takes constant time.
\end{proof}

We use algorithm \textsf{DerivedLowHigh}, shown below,
to order the vertices in $C'(z)$ according to
a low-high order of $\zeta$ of $G_A$.
After computing $G_A$, we construct two divergent spanning trees $B_A$ and $R_A$ of $G_A$.
For each vertex $v \not= z$, if $(z, v)$ is an edge of $G_A$,
we replace the parent of $v$ in $B_A$ and in $R_A$, denoted by $b_A(v)$ and $r_A(v)$, respectively, by $z$.
Then we use algorithm \textsf{AuxiliaryLowHigh} to compute a low-high order $\zeta$ of $G_A$.
Algorithm \textsf{AuxiliaryLowHigh} is
a slightly modified version of a linear-time algorithm of \cite[Section 6.1]{DomCert:TALG} to compute a low-high order.
Our modified version computes a low-high order $\zeta$ of $G_A$ that is an augmentation of $\delta_z$.
To obtain such a low-high order, we need to assign
to $\alpha^{\ast}$ the lowest number in $\zeta$ and to
$\beta^{\ast}$ the highest number in $\zeta$.
The algorithm works as follows.
While $G_A$ contains at least four vertices, we choose a vertex $v \not \in \{\alpha^{\ast}, \beta^{\ast} \}$  whose in-degree in $G_A$ exceeds its number of children in $B_A$ plus its number of children in $R_A$ and remove it from $G_A$.
(From this choice of $v$ we also have that $v \not= z$.)
Then we compute recursively a low-high order for the resulting flow graph, and insert $v$ in an appropriate location, defined by $b_A(v)$ and $r_A(v)$.

\begin{algorithm}[h!]

\LinesNumbered
\DontPrintSemicolon

Compute the derived affected flow graph $G_A= (V_A, E_A, z)$.\;

Compute two divergent spanning trees $B_A$ and $R_A$ of $G_A$.\;

\ForEach{vertex $v \in V_A \setminus \{z, \alpha^{\ast}, \beta^{\ast} \}$}{
\lIf{$\mathit{mark}(v) = \mathit{true}$}{set $b_A(v) \leftarrow z$ and $r_A(v) \leftarrow z$}}

Initialize a list of vertices $\Lambda \leftarrow \emptyset$.\;

Compute $\Lambda \leftarrow \mathsf{AuxiliaryLowHigh}(G_A,B_A,R_A,\Lambda)$.\;

Order the set of children $C'(z)$ of $z$ in $D'$ according to $\Lambda$.

\caption{\textsf{DerivedLowHigh}$(z,A,\mathit{mark})$}

\end{algorithm}

\begin{algorithm}[h!]
\LinesNumbered
\DontPrintSemicolon

\If{$G_A$ contains only three vertices}
{
set $\Lambda \leftarrow \langle \alpha^{\ast}, \beta^{\ast} \rangle$\; 
\KwRet $\Lambda$\;
}

Let $v \not\in \{\alpha^{\ast}, \beta^{\ast}\}$ be a vertex whose in-degree in $G_A$ exceeds its number of children in $B_A$ plus its number of children in $R_A$.\;

Delete $v$ and its incoming edges from $G_A$, $B_A$, and $R_A$.\;

\If{$v$ was not a leaf in $B_A$}
{
let $w$ be the child of $v$ in $B_A$; replace $b_A(w)$ by $b_A(v)$
}
\ElseIf{$v$ was not a leaf in $R_A$}
{
let $w$ be the child of $v$ in $R_A$; and replace $r_A(w)$ by $r_A(v)$
}

Call \textsf{AuxiliaryLowHigh}$(G_A,B_A,R_A,\Lambda)$ recursively for the new graph $G_A$.

\If{$b_A(v) = z$}
{
insert $v$ anywhere between $\alpha^{\ast}$ and $\beta^{\ast}$ in $\Lambda$
}
\Else
{
insert $v$ just before $b_A(v)$ in $\Lambda$ if $r_A(v)$ is before $b_A(v)$ in $\Lambda$, just after $b_A(v)$ otherwise
}
\KwRet $\Lambda$\;

\caption{\textsf{AuxiliaryLowHigh}$(G_A,B_A,R_A,\Lambda)$}
\end{algorithm}

\begin{lemma}
\label{lemma:auxiliary-low-high}
Algorithm \textsf{AuxiliaryLowHigh} is correct, that is, it computes a low-high order $\zeta$ of $G_A$, such that for all $v \in V_A \setminus \{z, \alpha^{\ast}, \beta^{\ast}\}$,
$\alpha^{\ast} <_{\zeta} v <_{\zeta} \beta^{\ast}$.
\end{lemma}
\begin{proof}
We first show that algorithm \textsf{AuxiliaryLowHigh} runs to completion, i.e., it selects every vertex $v \in V_A \setminus \{z, \alpha^{\ast}, \beta^{\ast}\}$ at some execution of line 5.
The recursive call in line 13 invokes algorithm \textsf{AuxiliaryLowHigh} on a sequence of smaller flow graphs $G_A$.
We claim that the following invariants hold for each such flow graph $G_A$:
\begin{itemize}
\item[(i)] the dominator tree $D_A$ of $G_A$ is flat;
\item[(ii)] the subgraphs $B_A$ and $R_A$ corresponding to $G_A$ are divergent spanning trees of $G_A$ rooted at $z$;
\item[(iii)] for every $v \not= z$, either $b_A(v) = r_A(v) = z$ or $b_A(v)$, $r_A(v)$, and $z$ are all distinct.
\end{itemize}
For the initial graph $G_A$ the invariants hold by construction.
Assume that the invariants hold on entry to line 5.
Suppose, now, that line 5 chooses a vertex $v \not\in \{\alpha^{\ast},\beta^{\ast}\}$. Since $v$ has in-degree at most $2$ in $G_A$, the choice of $v$ implies that it has at most one outgoing edge.
Hence $v$ is a leaf in either $B_A$ or $R_A$. If it is a leaf in both, deleting $v$ and its incoming edges preserves all the invariants.
Suppose $v$ is a leaf in $R_A$ but not $B_A$. Then $v$ has in-degree $2$ in $G_A$; that is, $b_A(v) \not= r_A(v)$, which implies by (iii) that $b_A(v)$, $r_A(v)$, and $v$ are distinct siblings in $D_A$.
Let $w$ be the child of $v$ in $B_A$.  Since $r_A(w) \not= v$, $v$, $r_A(w)$, and $z$ are distinct by (iii).
Also  $r_A(w) \not= b_A(v)$, since $r_A(w) = b_A(v)$ would imply that $r_A(w)$ dominates $w$ by (ii). Finally, $b_A(v)  \not= z$, since $b_A(v)$ is a sibling of $v$ and hence of $w$ in $D_A$.  We conclude that replacing $b_A(w)$ by $b_A(v)$ in line 8 preserves (iii).  This replacement preserves (i) since $v$ does not dominate $w$, it preserves (ii) since it removes $v$ from the path in $B_A$ from $s$ to $w$.  Replacing $b_A(w)$ makes $v$ a leaf in $B_A$, after which its deletion preserves (i)-(iii). 

Now we show that the invariants imply that line 5 can always choose a vertex $v$.
All vertices in $V_A \setminus z$ are leaves in $D_A$.
Let $X$ be the subset of $V_A$ that consists of the vertices $x$ such that $b_A(x) \not= r_A(x)$.
Each vertex in $X$ has in-degree $2$ in $G_A$, so there are $2|X|$ edges that enter a vertex in $X$.
By invariant (iii), each edge leaving a vertex in $X$ enters a vertex in $X$.
Invariant (iii) also implies that at least two edges enter $X$ from $V_A \setminus X$.
Hence, there are at most $2(|X|-1)$ edges that leave a vertex in $X$,
so there must be a vertex $v$ in $X$ with out-degree at most $1$.
We claim that $v$ can be selected in line 5. First note that the in-degree of $v$ in $G_A$ exceeds its out-degree in $G_A$.
If $v$ is a leaf in both $B_A$ and $R_A$ then it can be selected. If not, then $v$ must be a leaf in either $B_A$ or $R_A$, since
otherwise its common child $w$ in $B_A$ and $R_A$ would violate (ii). Hence $v$ can be selected in this case also.

Finally, we claim that the computed order is low-high for $G_A$, such that $\alpha^{\ast}$ is first and $\beta^{\ast}$ is last in this order.
The latter follows by the assignment in line 2. So the claim is immediate if $G_A$ has three vertices.
Suppose, by induction, that this is true if $G_A$ has $k \ge 3$ vertices. Let $G_A$ have $k + 1$ vertices and let $v$ be the vertex chosen for deletion.
The insertion position of $v$ guarantees that $v$ has the low-high property. All vertices in $G_A$ after the deletion of $v$ have the low-high property in the new $G_A \setminus z$ by the induction hypothesis, so they have the low-high property in the old $G_A$ with the possible exception of $w$, one of whose incoming edges differs in the old and the new $G_A$. Suppose $b_A(w)$ differs; the argument is symmetric if $r_A(w)$ differs.
Now we have that $v$, $w$, $b_A(v)$, and $r_A(w)$ are distinct children of $z$ in $D_A$. Since $w$ has the low-high property in the new $G_A$, it occurs in $\Lambda$ between $r_A(w)$ and $b_A(v)$.  Insertion of $v$ next to $b_A(v)$ leaves $w$ between $r_A(w)$ and $v$, so it has the low-high property in the old $G_A$ as well.
\end{proof}

The correctness of algorithm \textsf{InsertEdge} follows from Lemmata \ref{lemma:unaffected}, \ref{lemma:unaffected-children} and \ref{lemma:auxiliary-low-high}.

\begin{lemma}
\label{lemma:insert-edge-correct}
Algorithm \textsf{InsertEdge} is correct.
\end{lemma}
\begin{proof}
Let $(G', D', \delta', \mathit{mark}', \mathit{low}', \mathit{high}')$ be the output of  \textsf{InsertEdge}$(G, D, \delta, \mathit{mark}, \mathit{low}, \mathit{high}, e)$.
We only need to consider the case where both endpoints of the inserted edge $e=(x,y)$ are reachable in $G$. Let $A$ be the set of affected vertices, and let $z=\mathit{nca}(x,y)$.
Also, let $c$ be the child of $z$ in $D$ that is a common ancestor of all vertices in $A$.
We will show that the computed order $\delta'$ is a low-high order of $G'$ that agrees with $\delta$.
This fact implies that the arrays $\mathit{mark}'$, $\mathit{low}'$, $\mathit{high}'$ were updated correctly, since
their entries did not change for the vertices in $V \setminus \big ( A \cup \{c\} \big )$.

By construction, $\delta'$ agrees with $\delta$. Let $\delta_z$ (resp., $\delta'_z$) be the restriction of $\delta$ (resp., $\delta'$) to $C(z)$ (resp., $C'(z)$).
Then, by Lemma \ref{lemma:auxiliary-low-high}, $\delta'_z$ is an augmentation of $\delta_z$.
So, by Lemmata \ref{lemma:unaffected} and \ref{lemma:unaffected-children}, $\delta'$ is a low-high order in $G'$ for any vertex $v \not\in A \cup \{c\}$.
Finally, Lemma \ref{lemma:auxiliary-low-high} implies that $\delta'$ is also a low-high order in $G'$ for the vertices in $A \cup \{c\}$.
\end{proof}

\begin{theorem}
\label{theorem:IncLowHigh}
Algorithm \textsf{InsertEdge} maintains a low-high order of a flow graph $G$ with $n$ vertices through a sequence of edge insertions in
$O(mn)$ total time, where $m$ is the total number of edges in $G$ after all insertions.
\end{theorem}
\begin{proof}
Consider the insertion of an edge $(x,y)$. If $y$ was unreachable in $G$
then we compute $D$ and a low-high order in $O(m)$ time. Throughout
the whole sequence of $m$ insertions, such an event can happen $O(n)$ times,
so all insertions to unreachable vertices are handled in $O(mn)$ total time.

Now we consider the cost of executing \textsf{InsertEdge}
when both $x$ and $y$ are reachable in $G$.
Let $\nu$ be the number of scanned vertices, and let $\mu$ be the number of their adjacent edges.
We can update the dominator tree and locate the affected vertices (line 8) in $O(\nu+\mu+n)$ time~\cite{dyndom:2012}.
Computing $z=\mathit{nca}(x,y)$ in line 7 takes $O(n)$ time
just by using the parent function $d$ of $D$.
Lines 9--10 and 12 are also executed in $O(n)$ time.
The for loop in lines 13--16 takes $O(\nu+\mu)$ since we only need to examine the scanned edges. (Variables
$\mathit{low}(c)$ and $\mathit{high}(c)$ need to be updated only if there is a scanned edge entering $c$.)
It remains to account for time to compute $G_A$ and
a low-high order of it.
From Lemma \ref{lemma:derived-affected-graph}, the derived affected flow graph
can be constructed in $O(\nu+\mu)$ time. In algorithm \textsf{AuxiliaryLowHigh},
we represent the list $\Lambda$ with the off-line dynamic list maintenance data structure of \cite{DomCert:TALG},
which supports insertions (in a given location) and order queries in constant time.
With this implementation, \textsf{AuxiliaryLowHigh} runs in linear-time, that is $O(\nu+\mu)$.
So \textsf{InsertEdge} runs in $O(\nu+\mu+n)$ time.
The $O(n)$ term gives a total cost of $O(mn)$ for the whole sequence of $m$ insertions.
We distribute the remaining $O(\nu+\mu)$ cost to the scanned vertices and edges, that is $O(1)$ per scanned vertex or edge.
Since the depth in $D$ of every scanned vertex decreases by at least one, a vertex and an edge can be scanned at most
$O(n)$ times. Hence, each vertex and edge can contribute at most $O(n)$ total cost through the whole sequence of $m$ insertions.
The $O(m n)$ bound follows.
\end{proof}

\section{Applications of incremental low-high orders}
\label{sec:applications}

\subsection{Strongly divergent spanning trees and path queries}
\label{sec:divergent}

We can use the arrays $\mathit{mark}$, $\mathit{low}$, and $\mathit{high}$ 
to maintain a pair of strongly divergent spanning trees, $B$ and $R$,
of $G$ after each update.
Recall that $B$ and $R$ are \emph{strongly divergent} if for every pair of vertices $v$ and $w$, we have
$B[s,v] \cap R[s,w] = D[s,v] \cap D[s,w]$ or $R[s,v] \cap B[s,w] = D[s,v] \cap D[s,w]$.
Moreover, we can construct $B$ and $R$ so that they are also edge-disjoint except for the bridges of $G$. A \emph{bridge} of $G$ is an edge $(u, v)$ that is contained in every path from $s$ to $v$.
Let $b(v)$ (resp., $r(v)$) denote the parent of a vertex $v$ in $B$ (resp., $R$).
To update $B$ and $R$ after the insertion of an edge $(x,y)$, we only
need to update $b(v)$ and $r(v)$ for the affected vertices $v$, and possibly for their common ancestor $c$ that is
a child of $z=\mathit{nca}(x,y)$ from Lemma \ref{lemma:affected-child}.
We can update $b(v)$ and $r(v)$ of each vertex $v \in A \cup \{ c \}$ as follows:
set $b(v) \leftarrow d(v)$ if $\mathit{low}(v) = \mathit{null}$, $b(v) \leftarrow \mathit{low}(v)$ otherwise;
set $r(v) \leftarrow d(v)$ if $\mathit{high}(v) = \mathit{null}$,
$r(v) \leftarrow \mathit{high}(v)$ otherwise.
If the insertion of $(x,y)$ does not affect $y$, then $A = \emptyset$ but we may still need to update
$b(y)$ and $r(y)$ if $x \not\in D(y)$ in order to make $B$ and $R$ maximally edge-disjoint.
Note that in this case $z=d(y)$, so we only need to check if both $\mathit{low}(y)$ and
$\mathit{high}(y)$ are null. If they are, then we set $\mathit{low}(y) \leftarrow x$ if $x <_{\delta} y$,
and set $\mathit{high}(y) \leftarrow x$ otherwise. Then, we can update $b(y)$ and $r(y)$ as above.

Now consider a query that, given two vertices $v$ and $w$, asks for two maximally vertex-disjoint paths, $\pi_{sv}$ and $\pi_{sw}$, from $s$ to $v$ and from $s$ to $w$, respectively.
Such queries were used in \cite{Tholey2012} to give a linear-time algorithm for the $2$-disjoint paths problem on a directed acyclic graph.
If $v <_{\delta} w$, then we select $\pi_{sv} \leftarrow B[s,v]$ and $\pi_{sw} \leftarrow R[s,w]$; otherwise, we
select $\pi_{sv} \leftarrow R[s,v]$ and $\pi_{sw} \leftarrow B[s,w]$.
Therefore, we can find such paths in constant time, and output them in $O(|\pi_{sv}|+|\pi_{sw}|)$ time.
Similarly, for any two query vertices $v$ and $w$, we can report a path $\pi_{sv}$ from $s$ to $v$ that avoids $w$.
Such a path exists if and only if $w$ does not dominate $v$, which we can test in constant time
using the ancestor-descendant relation in $D$~\cite{domin:tarjan}.
If $w$ does not dominate $v$, then we select $\pi_{sv} \leftarrow B[s,v]$ if $v <_{\delta} w$, and select $\pi_{sv} \leftarrow R[s,v]$ if $w <_{\delta} v$.

\subsection{Fault tolerant reachability}
\label{sec:reachability}

Baswana et al.~\cite{FaultTolerantReachability} study the following reachability problem.
We are given a flow graph $G = (V, E, s)$
and a spanning tree $T =(V, E_T)$ rooted at $s$.
We call a set of edges $E'$ \emph{valid} if the subgraph $G' = (V, E_T \cup  E', s)$ of $G$ has the same dominators as $G$.
The goal is to find a valid set of minimum cardinality.
As shown in \cite{DomCert:TALG:Add}, we can compute a minimum-size valid set in $O(m)$ time, given the dominator tree $D$ and a low-high order of $\delta$ of it.
We can combine the above construction with our incremental low-high algorithm to solve the incremental version of the fault tolerant reachability problem,
where $G$ is modified by edge insertions and we wish to compute efficiently a valid set for any query spanning tree $T$. Let $t(v)$ be the parent of $v$ in $T$.
Our algorithm maintains, after each edge insertion, a low-high order $\delta$ of $G$, together with the $\mathit{mark}$, $\mathit{low}$, and $\mathit{high}$ arrays.
Given a query spanning tree $T=(V,E_T)$, we can compute a valid set of minimum cardinality $E'$ as follows.
For each vertex $v \not= s$, we apply the appropriate one of the following cases:
(a) If $t(v) = d(v)$ then we do not insert into $E'$ any edge entering $v$.
(b) If $t(v) \not= d(v)$ and $v$ is marked then we insert $(d(v),v)$ into $E'$.
(c) If $v$ is not marked then we consider the following subcases: If $t(v) >_{\delta} v$, then we insert into $E'$ the
edge $(x, v)$ with $x = \mathit{low}(v)$.
Otherwise, if $t(v) <_{\delta} v$, then we insert into $E'$ the edge $(x, v)$ with $x = \mathit{high}(v)$.
Hence, can update the minimum valid set in $O(mn)$ total time.

We note that the above construction can be easily generalized for the case where $T$ is forest, i.e., when $E_T$ is a subset of the edges
of some spanning tree of $G$. In this case, $t(v)$ can be null for some vertices $v \not =s$.
To answer a query for such a $T$, we apply the previous construction with the following modification when $t(v)$ is null.
If $v$ is marked then we insert $(d(v),v)$ into $E'$, as in case (b). Otherwise, we insert both edges entering $v$ from
$\mathit{low}(v)$ and $\mathit{high}(v)$.
In particular, when $E_T = \emptyset$, we compute a subgraph $G' = (V,  E', s)$ of $G$ with minimum number of edges that
has the same dominators as $G$. This corresponds to the case $k=1$ in \cite{FaultTolerantReachability:STOC16}.

\subsection{Sparse certificate for $2$-edge-connectivity}
\label{sec:2ec}

Let $G=(V,E)$ be a strongly connected digraph.
We say that vertices $u, v\in V$ are \emph{$2$-edge-connected} if there are two edge-disjoint directed paths from $u$ to $v$
and two edge-disjoint directed paths from $v$ to $u$. (A path from $u$ to $v$ and a path from $v$ to $u$ need not be edge-disjoint.)
A \emph{$2$-edge-connected block} of a digraph $G=(V,E)$ is defined as a maximal subset $B \subseteq V$ such that every two vertices in $B$ are $2$-edge-connected. If $G$ is not strongly connected, then its $2$-edge-connected blocks are the $2$-edge-connected blocks of each strongly connected component of $G$.
A \emph{sparse certificate} for the $2$-edge-connected blocks of $G$ is a spanning subgraph $C(G)$ of $G$ that has $O(n)$ edges
and maintains the same $2$-edge-connected blocks as $G$. Sparse certificates of this kind allow us to speed up computations, such as finding the actual edge-disjoint paths that connect a pair of vertices (see, e.g., \cite{sparse-k-connected:ni}).
The $2$-edge-connected blocks and a corresponding sparse certificate can be computed in $O(m+n)$ time~\cite{2ECB}.
An incremental algorithm for maintaining the $2$-edge-connected blocks is presented in \cite{GIN16:ICALP}. This algorithm maintains
the dominator tree of $G$, with respect to an arbitrary start vertex $s$, and of its reversal $G^R$, together with the auxiliary components of
$G$ and $G^R$, defined next.

Recall that an edge $(u,v)$ is a \emph{bridge} of a flow graph $G$ with start vertex $s$ if all paths from $s$ to $v$ include $(u,v)$.
After deleting from the dominator tree $D$ the bridges of $G$, we obtain the \emph{bridge decomposition} of $D$ into a forest $\mathcal{D}$.
For each
root $r$ of a tree in the bridge decomposition $\mathcal{D}$ we define the \emph{auxiliary graph $G_r = (V_r, E_r)$ of $r$} as follows.
The vertex set
$V_r$ of $G_r$ consists of all the vertices in $D_r$.
The edge set $E_r$ contains all the edges of $G$ among the vertices of $V_r$, referred to as
\emph{ordinary} edges, and a set of \emph{auxiliary} edges, which are obtained by contracting vertices in $V\setminus V_r$, as follows.
Let $v$ be a vertex in $V_r$ that has a child $w$ in $V \setminus V_r$. Note that $w$ is a root in the bridge decomposition $\mathcal{D}$ of $D$.
For each such child $w$ of $v$, we contract $w$ and all its descendants in $D$ into $v$.
The \emph{auxiliary components} of $G$ are the strongly connected components of each auxiliary graph $G_r$.

We sketch how to extend the incremental algorithm of \cite{GIN16:ICALP} so that it also maintains a sparse certificate $C(G)$ for the $2$-edge-connected components of $G$,
in $O(mn)$ total time. It suffices to maintain the auxiliary components in $G$ and $G^R$, and two maximally edge-disjoint divergent spanning trees for each of $G$ and $G^R$.
We can maintain these divergent spanning trees as described in Section \ref{sec:divergent}. To identify the auxiliary components, the algorithm of  \cite{GIN16:ICALP} uses, for each auxiliary graph, an incremental algorithm for maintaining strongly connected components~\cite{Bender:IncCycleDetection:TALG}.
It is easy to extend this algorithm so that it also computes $O(n)$ edges that define these strongly connected components.
The union of these edges and of the edges in the divergent spanning trees are the edges of $C(G)$.

\section{$2$-vertex-connected spanning subgraph}
\label{sec:2VCSS}

Let $G=(V,E)$ be a strongly connected digraph.
A vertex $x$ of $G$ is a \emph{strong articulation point} if $G \setminus x$ is not strongly connected.
A strongly connected digraph $G$ is \emph{$2$-vertex-connected} if it has at least three vertices and no strong articulation points~\cite{2vc,Italiano2012}.
Here we consider the problem of approximating a smallest $2$-vertex-connected spanning subgraph (\textsf{2VCSS}) of $G$. This problem is NP-hard~\cite{GJ:NP}.
We show that algorithm 
\textsf{LH-Z}
(given below), which uses low-high orders, achieves a linear-time $2$-approximation for this problem. The
best previous approximation ratio achievable in linear-time was $3$~\cite{2VCSS:Geo}, so we obtain a substantial improvement.
The best approximation ratio for \textsf{2VCSS} is $3/2$, and is achieved by the algorithm of Cheriyan and Thurimella~\cite{CT00} in $O(m^2)$ time,
or in $O(m\sqrt{n}+n^2)$ by a combination of \cite{CT00} and \cite{2VCSS:Geo}.
Computing small spanning subgraphs is of particular importance when dealing with large-scale graphs, e.g., with hundreds of million to billion edges. In this framework, one big challenge is to design linear-time algorithms, since algorithms with higher running times might be practically infeasible on today's architectures.
Let $G=(V,E)$ be a strongly connected digraph.
In the following, we denote by $G^R=(V,E^R)$ the \emph{reverse digraph} of $G$ that results from $G$ after reversing
all edge directions.

\begin{algorithm}[t]

\LinesNumbered
\DontPrintSemicolon
\KwIn{$2$-vertex-connected digraph $G=(V,E)$}
\KwOut{$2$-approximation of a smallest $2$-vertex-connected spanning subgraph $H=(V,E_H)$ of $G$}

Choose an arbitrary vertex $s$ of $G$ as start vertex.\;

Compute a strongly connected spanning subgraph $H=(V \setminus s, E_H)$ of $G \setminus s$.\;

Set $H \leftarrow (V, E_H)$.\;

Compute a low-high order $\delta$ of flow graph $G$ with start vertex $s$.\;

\ForEach{vertex $v \not= s$}
{
\If{there are two edges $(u, v)$ and $(w,v)$ in $E_H$ such that $u<_{\delta} v$ and $v <_{\delta} w$}
{
do nothing}
\ElseIf{there is no edge  $(u, v) \in E_H$ such that $u<_{\delta} v$} 
{
find an edge $e=(u, v) \in E$ with $u<_{\delta} v$\;
set $E_H \leftarrow E_H \cup \{ e \}$
}
\ElseIf{there is no edge $(w, v) \in E_H$ such that $v<_{\delta} w$} 
{
find an edge $e=(w, v) \in E$ with $v<_{\delta} w$ or $w=s$\;
set $E_H \leftarrow E_H \cup \{e\}$
}
}

Execute the analogous steps of lines 4--17 for the reverse flow graph $G^R$ with start vertex $s$.\;

\KwRet $H=(V,E_H)$

\caption{\textsf{LH-Z}$(G)$}
\end{algorithm}

\begin{lemma}
\label{lemma:LH2VCSS}
Algorithm \textsf{LH-Z} computes a $2$-vertex-connected spanning subgraph of $G$.
\end{lemma}
\begin{proof}
We need to show that the computed subgraph $H$ is $2$-vertex-connected.
From \cite{Italiano2012}, we have that a digraph $H$ is $2$-vertex connected if and only if it satisfies the following property:
For an arbitrary start vertex $s \in V$,
flow graphs $H=(V,E,s)$ and $H^R=(V,E^R,s)$ have flat dominator trees, and $H \setminus s$ is strongly connected.
The digraph $H$ computed by algorithm \textsf{LH-Z} satisfies the latter condition because of line 2.
It remains to show that $H$ has flat dominator tree. The same argument applies for $H^R$, thus completing the proof.
Let $\delta$ be the low-high order $\delta$ of $G$, computed in line $3$.
We argue that after the execution of the for loop in lines 5--17, $\delta$ is also a low-high order for all vertices in $H$.
Consider an arbitrary vertex $v \not= s$. Let $(x,v)$ be an edge entering $v$ in the strongly connected spanning subgraph of $G$ computed in line 2.
If $x >_{\delta} v$, then, by the definition of $\delta$,  there is at least one edge $(y,v) \in E$ such that $y <_{\delta} v$.
Hence, after the execution of the for loop for $v$, the edge set $E_H$ will contain at least two edges $(u,v)$ and $(w,v)$
such that $u <_{\delta} v <_{\delta} w$. On the other hand, if $x <_{\delta} v$, then the definition of $\delta$ implies that
there an edge $(y,v) \in E$ such that $y >_{\delta} v$ or $y = s$.  Notice that in either case $y \not= x$.
So, again, after the execution of the for loop for $v$, the edge set $E_H$ will contain at least two edges $(u,v)$ and $(w,v)$
such that either $u <_{\delta} v <_{\delta} w$, or $u <_{\delta} v$ and $w = s$.
It follows that $\delta$ is a low-high order for all vertices $v \not=s$ in $H$.
By \cite{DomCert:TALG}, this means that $H$ contains two strongly divergent spanning trees $B$ and $R$ of $G$.
Since $G$ has flat dominator tree, we have that $B[s,v] \cap R[s,v] = \{s,v\}$ for all $v \in V \setminus s$.
Hence, since $H$ contains $B$ and $R$, the dominator tree of $H$ is flat.
\end{proof}

We remark that the construction of $H$ in algorithm
\textsf{LH-Z} guarantees that $s$ will have in-degree and out-degree at least $2$
in $H$. (This fact is implicit in the proof of  Lemma \ref{lemma:LH2VCSS}.)
Indeed, $H$ will contain the edges from $s$ to the vertices in $V \setminus s$ with minimum and maximum order in $\delta$,
and the edges entering $s$ from the vertices in $V \setminus s$ with minimum and maximum order in $\delta^R$.

\begin{theorem}
\label{theorem:LH2VCSS}
Algorithm \textsf{LH-Z} computes a $2$-approximation for \textsf{2VCSS} in linear time.
\end{theorem}
\begin{proof}
We establish the approximation ratio of \textsf{LH-Z} by showing that $|E_H| \le 4n$. The approximation ratio of $2$ follows from the fact that any vertex in a $2$-vertex-connected digraph must have in-degree at least two.
In line 2 we can compute an approximate smallest strongly connected spanning subgraph of $G \setminus s$~\cite{KhullerSICOMP95}.
For this, we can use the linear-time algorithm of Zhao et al.~\cite{ZNI:MSCS:2003}, which selects at most $2(n-1)$ edges.
Now consider the edges selected in the for loop of lines 5--17. Since after line 2, $H \setminus s$ is strongly connected, each vertex $v \in V \setminus s$ has at least
one entering edge $(x,v)$. If $x <_{\delta} v$ then lines 10--11 will not be executed; otherwise, $v <_{\delta} x$ and lines 14--15 will not be executed.
Thus, the for loop of lines 5--17 adds at most one edge entering each vertex $v \not =s$.
The same argument implies that the analogous steps executed for $G^R$ add at most one edge leaving each vertex $v \not =s$.
Hence, $E_H$ contains at most $4(n-1)$ at the end of the execution.
\end{proof}

\begin{table}[t!]
\centering
\begin{scriptsize}
\begin{tabular}{|l|rrr|rrr|l|}
\hline
\textbf{Graph}                 & \multicolumn{3}{c|}{\textbf{Largest SCC}} & \multicolumn{3}{c|}{\textbf{2VCCs}} & \textbf{Type}                \\
                     & $n$              & $m$       & avg. $\delta$       & $n$          & $m$      & avg. $\delta$     &                     \\
\hline \hline
rome99                & 3352   & 8855    & 2.64  & 2249   & 6467    & 2.88  & road network        \\
\hline twitter-higgs-retweet & 13086  & 63537   & 4.86  & 1099   & 9290    & 8.45  & twitter             \\
\hline enron                 & 8271   & 147353  & 17.82 & 4441   & 123527  & 27.82 & enron mails         \\
\hline web-NotreDame         & 48715  & 267647  & 5.49  & 1409   & 6856    & 4.87  & web                 \\
                      &        &         &       & 1462   & 7279    & 4.98  &                     \\
                      &        &         &       & 1416   & 13226   & 9.34  &                     \\
\hline soc-Epinions1         & 32220  & 442768  & 13.74 & 17117  & 395183  & 23.09 & trust network       \\
\hline Amazon-302            & 241761 & 1131217 & 4.68  & 55414  & 241663  & 4.36  & co-purchase         \\
\hline WikiTalk              & 111878 & 1477665 & 13.21 & 49430  & 1254898 & 25.39 & Wiki communications \\
\hline web-Stanford          & 150475 & 1576157 & 10.47 & 5179   & 129897  & 25.08 & web                 \\
                      &        &         &       & 10893  & 162295  & 14.90 &                     \\
\hline web-Google            & 434818 & 3419124 & 7.86  & 77480  & 840829  & 10.85 & web                 \\
\hline Amazon-601            & 395230 & 3301051 & 8.35  & 276049 & 2461072 & 8.92  & co-purchase         \\
\hline web-BerkStan          & 334857 & 4523232 & 13.51 & 1106   & 8206    & 7.42  & web                 \\
                      &        &         &       & 4927   & 28142   & 5.71  &                     \\
                      &        &         &       & 12795  & 347465  & 27.16 &                     \\
                      &        &         &       & 29145  & 439148  & 15.07 &                    \\

\hline

\end{tabular}
\end{scriptsize}
\caption{Real-world graphs used in the experiments, sorted by the file size of their largest SCC. We used both the largest SCC and the some of the 2VCCs (inside the largest SCC) in our experiments.\label{tab:dataset}}
\end{table}

\section{Empirical Analysis}
\label{sec:experimental}

For the experimental evaluation we use the graph datasets shown in Table~\ref{tab:dataset}. We wrote our implementations in {\tt C++}, using {\tt g++ v.4.6.4} with full optimization (flag {\tt -O3}) to compile the code. We report the running times on a GNU/Linux machine, with Ubuntu (12.04LTS): a Dell PowerEdge R715 server 64-bit NUMA machine with four AMD Opteron 6376 processors  and 128GB of RAM memory. Each processor has 8 cores sharing a 16MB L3 cache, and each core has a 2MB private L2 cache and 2300MHz speed.
In our experiments we did not use any parallelization, and each algorithm ran on a
single core.
We report CPU times measured with the \texttt{getrusage} function, averaged over ten different runs.
In Table~\ref{tab:dataset} we can see some statistics about the real-world graphs we used in our experimental evaluation.

\begin{figure}[t!]
\centering
\includegraphics[width=0.8\textwidth, trim=0cm 0cm 0cm 4.0cm, clip]{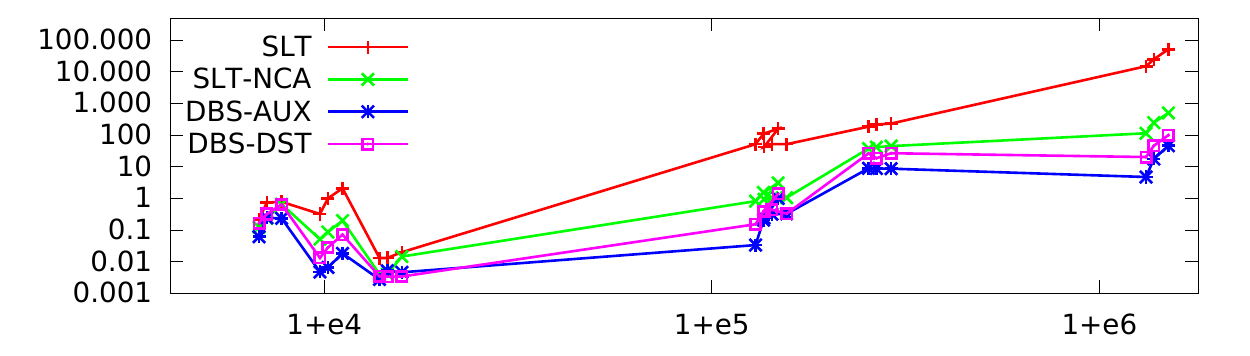} \\
\includegraphics[width=0.8\textwidth, trim=0cm 0cm 0cm 4.0cm, clip]{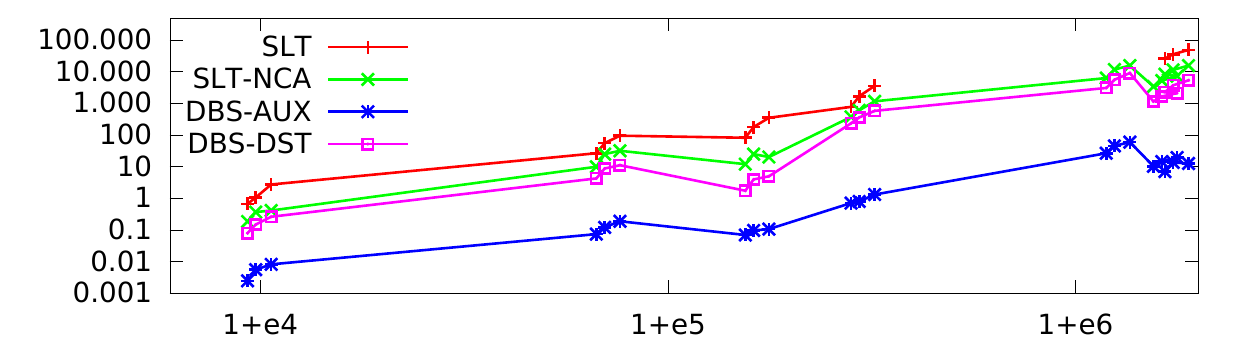}%
\caption{Incremental low-high order: dynamized 2VC graphs (top) and edge insertion in strongly connected graphs (bottom). Running times, in seconds, and number of edges both shown in logarithmic scale.
\label{fig:incr}}
\vspace{-0.4cm}
\end{figure}

\subsection{Incremental low-high order.}
\label{sec:exp-low-high}
We compare the performance of four algorithms. As a baseline, we use a static low-high order algorithm from \cite{DomCert:TALG} based on
an efficient implementation of
the Lengauer-Tarjan algorithm for computing dominators~\cite{domin:lt} from \cite{dom_exp:gtw}.
Our baseline algorithm, \textsf{SLT}, constructs, as intermediary, two divergent spanning trees. 
After each insertion of an edge $(x,y)$, \textsf{SLT} recomputes a low-high order if $x$ is reachable. An improved version of this algorithm, that we refer to as \textsf{SLT-NCA}, tests if the insertion of $(x,y)$ affects the dominator tree by computing the nearest common ancestor of $x$ and $y$. If this is the case, then \textsf{SLT-NCA} recomputes a low-high order as \textsf{SLT}. The other two algorithms are the ones we presented in Section \ref{sec:incremental}. For our simple algorithm, \textsf{DBS-DST}, we extend the incremental dominators algorithm \textsf{DBS} of \cite{dyndom:2012} with the computation of two divergent spanning trees and a low-high order, as in \textsf{SLT}. Algorithm \textsf{DBS-DST} applies these computations on a sparse subgraph of the input digraph that maintains the same dominators.
Finally, we tested an implementation of our more efficient algorithm, \textsf{DBS-AUX}, that updates the low-high order by computing a local low-high order of an auxiliary graph.

We compared the above incremental low-high order algorithms in two different field tests. In the first one, we considered 2-vertex connected graphs, and we dynamized them in the following manner: we removed a percentage of edges (i.e., 5\%, 10\%, and 20\% respectively), selected uniformly at random, that were incrementally added to the graph.
Note that during the execution of the algorithms some vertices may be unreachable at first. Also, at the end of all insertions, the final graph has flat dominator tree.
In Figure~\ref{fig:incr} (top) we can see that the algorithms are well distinguished: our \textsf{DBS-AUX} performs consistently better than the other ones (with the exception of two NotreDame instances). 
The total running times are given in Table \ref{table1}.
On average, \textsf{DBS-DST} is about $2.84$ times faster than \textsf{SLT-NCA}, with their relative performance depending on the density of the graph (the higher the average degree the better \textsf{DBS-DST} performs w.r.t. \textsf{SLT-NCA}.)
As we mentioned, the naive \textsf{SLT} is the worst performer.
The above observed behavior of the algorithms is similar also in the second test. Here, we consider the strongly connected graphs, and we incrementally insert random edges up to a certain percentage of the original number of edges (i.e., as before, 5\%, 10\%, and 20\% respectively). We use strongly connected graphs only in order to guarantee that all vertices are reachable from the selected source. (Strong connectivity has no other effect in these tests.)
The endpoints of each new edge are selected uniformly at random, and the edge is inserted if it is not a loop and is not already present in the current graph.
The ranking of the algorithms does not change, as we can see in Figure~\ref{fig:incr} (bottom), but the difference is bigger: we note a bigger gap of more than two orders of magnitude, in particular, between \textsf{DBS-AUX} and the couple \textsf{SLT-NCA} and \textsf{DBS-DST}.
This is expected because, unlike the first test, here all edges connect already reachable vertices. This means that \textsf{DBS-DST} and \textsf{DBS-AUX} do not execute a full restart for any of these insertions.
The total running times are given in Table \ref{table2}.

\begin{figure}[t!]
\centering
\begin{tabular}{r}
\vspace{-0.3cm}
\includegraphics[width=0.8\textwidth, trim=0cm 0cm 0cm 4.0cm, clip]{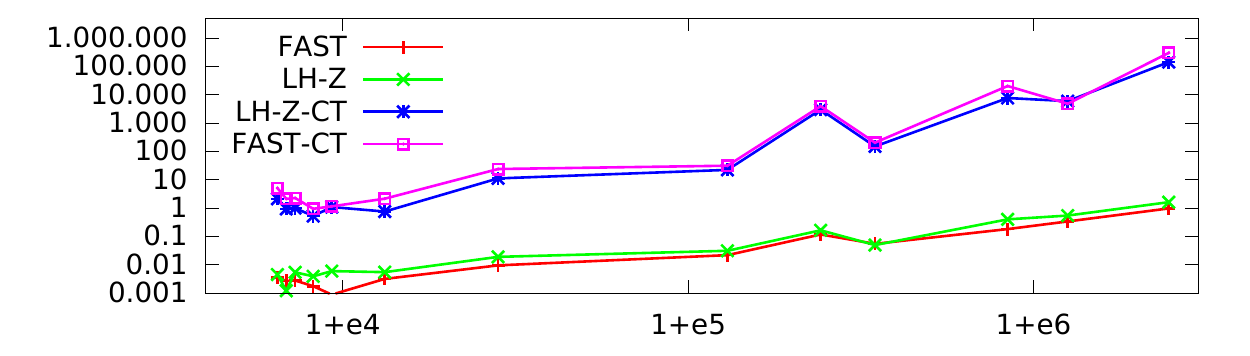} \\

\includegraphics[width=0.73\textwidth, trim=0cm 0cm 0cm 4.0cm, clip]{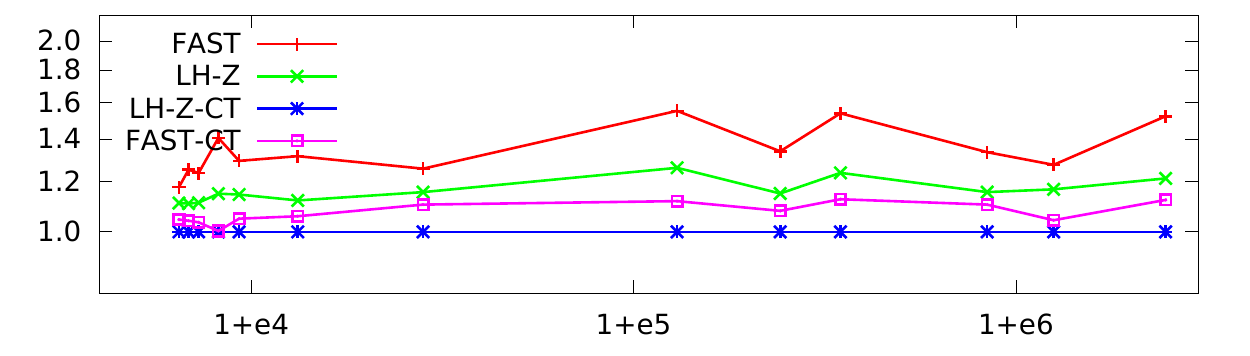}%
\vspace{-0.4cm}
\end{tabular}

\caption{Smallest 2-vertex-connected spanning subgraph. Top: running times, in seconds, and number of edges both shown in logarithmic scale. Bottom: relative size of the resulting 2VCSS.
\label{fig:a2vcc}}
\end{figure}

\begin{table}[]
\centering
\begin{scriptsize}
\begin{tabular}{lrrrrrrr}
\hline \textbf{Graph}         & \textbf{nodes} & \textbf{starting edges} & \textbf{final edges} & \textbf{SLT}      & \textbf{SLT-NCA}  & \textbf{DBS-AUX}  & \textbf{DBS-DST}  \\
\hline rome05        & 2249  & 6144           & 6467        & 0.216091 & 0.120026 & 0.060632 & 0.16457  \\
rome10        & 2249  & 5820           & 6467        & 0.734963 & 0.231678 & 0.242326 & 0.319025 \\
rome20        & 2249  & 5174           & 6467        & 0.772791 & 0.646389 & 0.231463 & 0.639627 \\
\hline twitter05     & 1099  & 8826           & 9290        & 0.320313 & 0.051123 & 0.004682 & 0.012953 \\
twitter10     & 1099  & 8361           & 9290        & 0.996879 & 0.085982 & 0.006498 & 0.027945 \\
twitter20     & 1099  & 7432           & 9290        & 2.06744  & 0.198226 & 0.018012 & 0.070981 \\
\hline NotreDame05   & 1416  & 12565          & 13226       & 0.012942 & 0.003981 & 0.002727 & 0.003421 \\
NotreDame10   & 1416  & 11903          & 13226       & 0.012958 & 0.003997 & 0.005094 & 0.003341 \\
NotreDame20   & 1416  & 10581          & 13226       & 0.019733 & 0.01446  & 0.004571 & 0.003374 \\
\hline enron05       & 4441  & 117351         & 123527      & 51.5453  & 0.811483 & 0.033019 & 0.152272 \\
enron10       & 4441  & 111174         & 123527      & 109.719  & 1.5252   & 0.204307 & 0.388753 \\
enron20       & 4441  & 98822          & 123527      & 158.83   & 3.08813  & 0.999617 & 1.40979  \\
\hline webStanford05 & 5179  & 123402         & 129897      & 42.1674  & 0.936905 & 0.236135 & 0.370119 \\
webStanford10 & 5179  & 116907         & 129897      & 51.2838  & 1.02925  & 0.316648 & 0.439147 \\
webStanford20 & 5179  & 103918         & 129897      & 51.6162  & 1.04364  & 0.323679 & 0.329067 \\
\hline Amazon05      & 55414 & 229580         & 241663      & 185.868  & 37.4155  & 8.91418  & 26.5169  \\
Amazon10      & 55414 & 217497         & 241663      & 214.185  & 41.4565  & 8.80395  & 18.4656  \\
Amazon20      & 55414 & 193330         & 241663      & 230.7    & 44.8627  & 8.66914  & 26.7402  \\
\hline WikiTalk05    & 49430 & 1192153        & 1254898     & 15026.2  & 113.946  & 4.7007   & 20.2353  \\
WikiTalk10    & 49430 & 1129408        & 1254898     & 24846.2  & 247.601  & 17.5997  & 45.7164  \\
WikiTalk20    & 49430 & 1003918        & 1254898     & 51682    & 500.581  & 45.9101  & 99.6058\\
\hline  
\end{tabular}
\end{scriptsize}
\caption{Running times of the plot shown in Figure~\ref{fig:incr} (top) \label{table1}.}

\end{table}

\begin{table}[]
\centering
\begin{scriptsize}
\begin{tabular}{lrrrrrrr}
\hline \textbf{Graph}         & \textbf{nodes} & \textbf{starting edges} & \textbf{final edges} & \textbf{SLT}      & \textbf{SLT-NCA}  & \textbf{DBS-AUX}  & \textbf{DBS-DST}  \\
\hline  rome05        & 3352   & 8855           & 9298        & 0.662001        & 0.185072 & 0.002457 & 0.07962  \\
rome10        & 3352   & 8855           & 9741        & 1.06533         & 0.366822 & 0.005531 & 0.147052 \\
rome20        & 3352   & 8855           & 10626       & 2.74201         & 0.410448 & 0.008154 & 0.259299 \\
\hline twitter05     & 13086  & 63537          & 66714       & 26.5965         & 9.94862  & 0.073755 & 4.2933   \\
twitter10     & 13086  & 63537          & 69891       & 55.4924         & 25.189   & 0.120719 & 8.73372  \\
twitter20     & 13086  & 63537          & 76244       & 96.3205         & 31.7239  & 0.186917 & 11.1431  \\
\hline enron05       & 8271   & 147353         & 154721      & 82.6084         & 11.9889  & 0.068994 & 1.72764  \\
enron10       & 8271   & 147353         & 162088      & 180.222         & 25.0999  & 0.09557  & 3.97011  \\
enron20       & 8271   & 147353         & 176824      & 353.174         & 20.1978  & 0.106017 & 4.94514  \\
\hline NotreDame05   & 48715  & 267647         & 281029      & 785.012         & 375.356  & 0.70628  & 234.757  \\
NotreDame10   & 48715  & 267647         & 294412      & 1691.05         & 610.29   & 0.79135  & 359.028  \\
NotreDame20   & 48715  & 267647         & 321176      & 3593.09         & 1168.5   & 1.31932  & 585.807  \\
\hline Amazon05      & 241761 & 1131217        & 1187778     & \textgreater24h & 6386.97  & 26.5493  & 3094.69  \\
Amazon10      & 241761 & 1131217        & 1244339     & \textgreater24h & 11905.7  & 45.1881  & 5628.51  \\
Amazon20      & 241761 & 1131217        & 1357460     & \textgreater24h & 15871    & 60.197   & 9157.5   \\
\hline WikiTalk05    & 111878 & 1477665        & 1551548     & \textgreater24h & 3414.28  & 10.3364  & 1157.84  \\
WikiTalk10    & 111878 & 1477665        & 1625432     & \textgreater24h & 5301.51  & 14.5151  & 1666.28  \\
WikiTalk20    & 111878 & 1477665        & 1773198     & \textgreater24h & 7296.72  & 19.5778  & 2124.6   \\
\hline webStanford05 & 150475 & 1576157        & 1654965     & \textgreater24h & 8403.03  & 7.028    & 2295.55  \\
webStanford10 & 150475 & 1576157        & 1733773     & \textgreater24h & 11503.4  & 13.7287  & 3749.12  \\
webStanford20 & 150475 & 1576157        & 1891388     & \textgreater24h & 15792.1  & 12.7093  & 5381.12\\
\hline  
\end{tabular}
\end{scriptsize}
\caption{Running times of the plot shown in Figure~\ref{fig:incr} (bottom) \label{table2}.}

\end{table}

\subsection{$2$-vertex-connected spanning subgraph.}
\label{sec:exp-2vcss}
In this experimental evaluation we compared four algorithms for computing the (approximated) smallest $2$-vertex-connected spanning subgraph.
Specifically, we tested two algorithms from \cite{2VCSS:Geo}, \textsf{FAST} which computes a $3$-approximation in linear-time by using divergent spanning trees, and \textsf{FAST-CT} which combines \textsf{FAST} with the $3/2$-approximation algorithm of Cheriyan and Thurimella~\cite{CT00}. In the experiments reported in \cite{2VCSS:Geo}, the former algorithm achieved the fastest running times, while the latter the best solution quality. We compare these algorithms against our new algorithm \textsf{LH-Z} of Section \ref{sec:2VCSS}, and a new hybrid algorithm \textsf{LH-Z-CT}, that combines
\textsf{LH-Z} with the algorithm of Cheriyan and Thurimella~\cite{CT00}.

Algorithm \textsf{LH-Z-CT} works as follows. First, it computes a $1$-matching $M$ in the input graph $G$~\cite{matching:GT}, using bipartite matching as in \cite{CT00}.
Let $H$ be the subgraph of $G \setminus s$, for arbitrary start vertex $s$, that contains only the edges in $M$. We compute the strongly connected components $C_1, \ldots, C_k$ in $H$, and
form a contracted version $G'$ of $G$ as follows. For each strongly connected component $C_i$ of $H$, we contract all vertices in $C_i$ into a representative vertex $u_i \in C_i$.
(Contractions are performed by union-find~\cite{dsu:tarjan} and merging lists of out-edges of $G$.)
Then, we execute the linear-time algorithm of Zhao et al.~\cite{ZNI:MSCS:2003} to compute a strongly connected spanning subgraph of $G'$,
and store the original edges of $G$ that correspond to the selected edges by the Zhao et al. algorithm. Let $Z$ be this set of edges.
We compute a low-high order of $G$ with root $s$, and use it in order to compute a 2-vertex-connected spanning subgraph $W$ of $G$ using as many edges from $Z$ and $M$ as possible, as in \textsf{LH-Z}.
Then, we run the filtering phase of Cheriyan and Thurimella. For each edge $(x,y)$ of $W$ that is not in $M$, we test if $x$ has two vertex-disjoint paths to $y$
in $W \setminus (x,y)$. If it does, then we set $W \leftarrow W \setminus (x,y)$. We remark that, similarly to \textsf{FAST-CT}, \textsf{LH-Z}
preserves the $3/2$ approximation guarantee of the Cheriyan-Thurimella algorithm for $k=2$ and improves its running time from
$O(m^2)$ to $O(m\sqrt{n}+n^2)$, for a digraph with $n$ vertices and $m$ arcs.
In our implementation, the bipartite matching is computed via max-flow, using an implementation of the Goldberg-Tarjan push-relabel algorithm~\cite{GT:maxflow} from \cite{DGRW:GraphPartition}, which is very fast in practice. (This implementation was provided by the authors of \cite{DGRW:GraphPartition}.)

In Figure~\ref{fig:a2vcc} (top) we can see the running times of the four algorithms. 
(See also Table \ref{table3}.)
It is easy to observe that the algorithms belong to two distinct classes, with \textsf{FAST} and \textsf{LH-Z} being faster than the other two by approximately five orders of magnitude.
In the bottom part of Figure~\ref{fig:a2vcc} we can see the relative size of the smallest spanning subgraph computed by the four algorithms. In all of our experiments, the smallest subgraph was the one computed by our new hybrid algorithm \textsf{LH-Z-CT}. One the other hand, on average \textsf{LH-Z} is only twice as slow as \textsf{FAST} but improves the solution quality by more than $13\%$.
Summing up, if one wants a fast and good solution \textsf{LH-Z} is the right choice.

\begin{table}[]
\centering
\begin{scriptsize}
\begin{tabular}{|l|rr|rr|rr|rr|rr|}
\hline \multicolumn{1}{|c|}{\textbf{Graph}} & \multicolumn{2}{c|}{\textbf{Size}}                              & \multicolumn{2}{c|}{\textbf{FAST}}                             & \multicolumn{2}{c|}{\textbf{LH-Z}}                             & \multicolumn{2}{c|}{\textbf{LH-Z-CT}}                          & \multicolumn{2}{c|}{\textbf{FAST-CT}}                          \\
                          & \multicolumn{1}{l}{nodes} & \multicolumn{1}{l|}{edges} & \multicolumn{1}{l}{time} & \multicolumn{1}{l|}{edges} & \multicolumn{1}{l}{time} & \multicolumn{1}{l|}{edges} & \multicolumn{1}{l}{time} & \multicolumn{1}{l|}{edges} & \multicolumn{1}{l}{time} & \multicolumn{1}{l|}{edges} \\
\hline rome99                    & 2249                      & 6467                      & 0.003581                 & 5691                      & 0.004513                 & 5370                      & 2.099891                 & 4837                      & 5.042213                 & 5057                      \\
web-NotreDame$_1$         & 1409                      & 6856                      & 0.002697                 & 3796                      & 0.001211                 & 3356                      & 0.927683                 & 3029                      & 2.119423                 & 3153                      \\
web-NotreDame$_2$         & 1462                      & 7279                      & 0.002784                 & 3949                      & 0.005436                 & 3545                      & 0.999031                 & 3189                      & 2.262008                 & 3300                      \\
web-BerkStan$_1$          & 1106                      & 8206                      & 0.001753                 & 3423                      & 0.00393                  & 2795                      & 0.525704                 & 2433                      & 0.949573                 & 2440                      \\
twitter-higgs-retweet     & 1099                      & 9290                      & 0.000885                 & 3553                      & 0.005976                 & 3143                      & 1.087169                 & 2745                      & 1.153608                 & 2879                      \\
web-NotreDame$_3$         & 1416                      & 13226                     & 0.003182                 & 4687                      & 0.005515                 & 3990                      & 0.751014                 & 3560                      & 2.143551                 & 3768                      \\
web-BerkStan$_2$          & 4927                      & 28142                     & 0.009555                 & 13391                     & 0.018985                 & 12296                     & 11.223287                & 10646                     & 23.887951                & 11750                     \\
web-Stanford              & 5179                      & 129897                    & 0.022056                 & 17940                     & 0.031001                 & 14583                     & 22.141856                & 11556                     & 31.346059                & 12920                     \\
Amazon-302                & 55414                     & 241663                    & 0.11804                  & 164979                    & 0.164081                 & 141467                    & 2986.022813              & 123095                    & 3935.135495              & 132847                    \\
web-BerkStan$_3$          & 12795                     & 347465                    & 0.055141                 & 45111                     & 0.049862                 & 36328                     & 149.561649               & 29307                     & 203.913794               & 32989                     \\
web-Google                & 77480                     & 840829                    & 0.182113                 & 256055                    & 0.401427                 & 221327                    & 7668.066338              & 191616                    & 20207.08225              & 211529                    \\
WikiTalk                  & 49430                     & 1254898                   & 0.338172                 & 176081                    & 0.548573                 & 161128                    & 5883.002974              & 138030                    & 4770.705853              & 143958                    \\
Amazon-601                & 276049                    & 2461072                   & 0.977274                 & 932989                    & 1.607812                 & 744345                    & 140894.0119              & 612760                    & 298775.2092              & 688159\\
\hline                    
\end{tabular}
\end{scriptsize}
\caption{Running times and number of edges in the resulting $2$-vertex-connected spanning subgraph;  plots shown in Figure~\ref{fig:a2vcc} \label{table3}.}
\end{table}

\vspace{0.5cm}

\noindent{\bf Acknowledgement.} We would like to thank Bob Tarjan for valuable comments and suggestions.

\bibliographystyle{plain}

\end{document}